\documentclass[12pt, draftclsnofoot, onecolumn]{IEEEtran}
\usepackage{amssymb}
\usepackage{amsfonts}
\usepackage{amsmath}
\usepackage{graphicx}
\usepackage{array}
\usepackage{psfrag}
\usepackage{cite}
\usepackage{mathtools}
\usepackage{color}
\usepackage{suffix}
\usepackage[normalem]{ulem}
\usepackage{balance}
\usepackage{amsthm}
\usepackage[english]{babel}
\usepackage{amsmath,esint}
\usepackage[mathscr]{euscript}
\usepackage{subfigure}
\usepackage{bigints}
\usepackage{nomencl}
\usepackage{balance}
\usepackage{multirow}
\usepackage{bm}
\usepackage{tipa}
\usepackage{ulem}
\makenomenclature
\subfiguretopcaptrue

\theoremstyle{definition}

\theoremstyle{plain}

\newtheorem{theorem}{Theorem}
\newtheorem{corollary}[]{Corollary}

\begin{document}

\title{On the Terminal Location Uncertainty in Elliptical Footprints: Application in Air-to-Ground Links}
\author{Alexander~Vavoulas,~Nicholas Vaiopoulos,~Harilaos~G.~Sandalidis,~and Konstantinos K. Delibasis

\thanks{The authors are with the Department of Computer Science and Biomedical
Informatics, University of Thessaly, Papasiopoulou 2-4, 35 131, Lamia,
Greece, e-mails: \{vavoulas,nvaio,sandalidis,kdelimpasis\}@dib.uth.gr.}
}
\maketitle

\begin{abstract}
Wireless transmitters (Txs) that radiate downward in a direction often generate circular footprints on the ground.   The configurational flexibility of these footprints is inherently limited as coverage adjustments are restricted to variations in radius, the only parameter available for tuning. This simplification is inadequate for scenarios that require asymmetric coverage, extended service areas, or dynamic footprint adaptation due to antenna tilt or changes in altitude of unmanned aerial vehicles (UAVs)  . In specific scenarios, the use of elliptical cells can offer increased flexibility for providing user coverage due to unique network characteristics. For example, an elliptical footprint can be produced when a practical directional antenna with unequal azimuth and elevation half-power beamwidths is used in high-speed railway networks. Another common scenario involves the production of an elliptical footprint when an airborne Tx radiates at an angle by tilting its directional antenna by a few degrees. This paper aims to investigate for the first time the association between the random location of the user within an elliptical coverage area and the performance of a wireless communication link considering these scenarios. We assume a UAV as a Tx, although a tall cellular base station tower could also be employed without losing generality.  To gain a deeper understanding of the impact of random location, we derive the relevant distance and signal-to-noise ratio metrics and examine the outage probability for a single-user link, as well as the throughput in a multiuser scenario. This analysis accounts for both random terminal locations and fading impairments in both cases. The findings provide valuable insights into the performance of comparable wireless systems.

\end{abstract}


\section{Introduction}
\subsection{Preliminaries}

\IEEEPARstart{O}{n} several occasions, wireless transmitters (Txs) at a certain altitude radiate directionally downward, generating a two-dimensional (2D) circular footprint on the ground. This consideration is often adopted in many studies exploring the performance in both outdoor, e.g., terrestrial cellular networks \cite{B:Goldsmith}, satellite-to-earth transmission \cite{B:Maral}, or unmanned aerial vehicle (UAV) systems \cite{B:Zhang}, and indoor environments with respect to visible light communications (VLC) inside a room \cite{B:Ghassemlooy}. Being a relatively simplistic shape, a circular disk makes the analysis much more straightforward, evident in problems focusing on the effect of the user's position uncertainty. As terrestrial users can be found at any point within the coverage area, identifying their location helps to ensure fast and seamless transmission, which is a prerequisite for the proper operation of the next generations of networks. However, the randomness of the terminal position also induces stochasticity in the path loss, thus introducing a new impairment that should be considered when assessing link quality. Therefore, the more realistic the coverage area, the more robust and reliable the performance analysis becomes from a mathematical point of view. 

In practice, there are several scenarios where circular disks do not accurately describe coverage areas. This limitation is evident when directional transmission occurs, leading to the formation of circular sectors \cite{J:Okorafor}. Furthermore, the hexagonal cell layout is often favored over circular disks for cellular networks since it provides a comprehensive coverage of a geographical region without any gaps \cite{B:Goldsmith}.  Although circular footprints simplify mathematical analysis, they often do not capture the asymmetric nature of real-world coverage areas. For example, in high-speed railway (HSR) networks or when directional antennas are tilted, the coverage area becomes elongated, rendering circular assumptions inadequate. This mismatch can lead to overestimations or underestimations of coverage and performance metrics, particularly in scenarios requiring precise spatial modeling, such as UAV communications \cite{J:Malikov}.  

What is of particular interest is the case of the elliptical shape, which is more complicated.  In contrast to circular models with uniform radial symmetry, the elliptical approach captures spatial anisotropy, which can influence coverage performance, signal strength distribution, and link reliability.  More precisely, this has proven to suitably represent narrow-strip-shaped cells in HSR communications. Elliptical cells offer a higher percentage of coverage than traditional circular ones since users are almost deployed along the rail tracks \cite{C:Liu}, \cite{J:Lin}. An elliptical rather than circular footprint can also be created in vertical transmission when the azimuth and elevation half-power beamwidths (HPBWs) of the directional antenna are not equal \cite{J:Noh}, \cite{C:Tamo}. In this case, the Tx projection on the ground coincides with the ellipse's center. Another more realistic scenario concerns an airborne Tx transmitting not vertically to Earth but at an angle, having turned its directional antenna by a few degrees. Due to the radiation pattern, an elliptical coverage area is formed centered at a distance from the projection of Tx on the terrain \cite{J:Azari}.

\subsection{Related Work}
Typically, terrestrial users are randomly dispersed within the coverage area of a wireless network. Therefore, determining a terminal position randomly placed in a specific 2D area and evaluating associated metrics, such as the distance from a particular point, is critical for wireless system performance. This problem is a research subject of stochastic geometry, as discussed in \cite{B:Chiu} and \cite{J:Hmamouche}. In this context, Mathai's book lists some interesting geometrical probability problems when considering ordinary topologies \cite{B:Mathai}. Some novel distance distributions associated with rhombuses and hexagons were deduced in \cite{B:Zhuang}. 

In addition, many studies focus on the typical circular disk shape in the communications field, as it simplifies mathematical analysis. However, this approach may not always accurately reflect real-world scenarios. For example, Gupta \textit{et al.} in \cite{J:Gupta}, analyzed a cascaded free space optical (FSO)-VLC system where end-to-end users were uniformly distributed in the indoor VLC cells. Furthermore, in \cite{J:Christopoulou}, an underwater downlink setup was investigated, where a Tx at a fixed location above the sea surface communicates with a receiver (Rx) randomly positioned below sea level. In the context of UAV-based communications, Vaiopoulos \textit{et. al.} in \cite{J:Vaiopoulos3}, explored a UAV-based broadcast scenario, assuming random locations for terrestrial terminals. Babu \textit{et al.} addressed a three-dimensional (3D) arrangement of aerial access points designed for energy-efficient uplink communication, considering intercell interference and the energy consumption of these access points \cite{J:Babu20}. Furthermore, Lai automated the determination of the appropriate number, positions, and altitudes of deployed UAVs to meet dynamic traffic demands within circular footprints \cite{J:Lai20}. Furthermore, Nafees \textit{et al.} in \cite{J:Nafees21}, introduced a UAV-based cellular network design employing multi-tier variable-size circle packing to optimize network coverage. Finally, Qureshi and Imran considered the tradeoff between the circular coverage area and the tilt angle of the base station for UAV communications in \cite{J:Qureshi19}.

Following the explosion of research on wireless communications, there has also been an emphasis on studying similar problems in three dimensions. In \cite{J:Talgat}, Talgat \textit{et al.} focused on the closest neighbor distances for satellites scattered on spherical surfaces following the binomial point process, whereas the performance of uniformly distributed satellites in a low-earth orbit network was considered in \cite{J:Okati}. Moreover, in \cite{J:Vaiopoulos} and \cite{J:Vaiopoulos2}, the authors evaluated a UAV and VLC scenario, accordingly, where an Rx is randomly located within a truncated conical volume. An analysis of underwater optical communications involving a conical spherical volume was also formed in \cite{J:Vaiopoulos4}. The exploration of wireless connectivity in the 3D space for a swarm of UAVs was carried out in \cite{J:Yuan20}, while \cite{J:Armeniakos20} concentrated on a specific scenario involving a finite number of UAVs uniformly distributed within a sphere.

Despite the presence of elliptical coverage areas, as mentioned previously, studies on the effect of the random position of users within ellipses are still lacking. Tracing back to the relevant literature, some mathematical calculations for the elliptical surface were presented in the old days in studies from the former Soviet Union. More precisely, in \cite{J:Geciauskas}, Ge\v{c}iauscas deduced the distribution function of a distance between two random points in the oval, which can be considered as a more generic shape than the ellipse, while in \cite{J:Geciauskas2}, he found the distribution function of a distance between a uniformly distributed point within the oval and a random point on its contour. In \cite{J:Piefke}, Piefke derived the chord length distribution of the ellipse, and in \cite{B:Parry}, Parry delved into the probability that a given distance separates two uniformly distributed points. The angle and time of arrival statistics for an elliptical scattering model were investigated in \cite{J:Ertel}. In \cite{C:Tamo}, an analytical method was proposed to generate uniformly distributed random locations for ground terminals within elliptical beam footprints. In a recent study, the work by \cite{J:Malikov23} focused on determining the smallest ellipse on the ground that encompasses the targets and establishing the corresponding tilted circle in the air. 

\subsection{Motivation}

The adoption of elliptical coverage patterns can be widely observed in several scenarios in typical wireless systems. A practical UAV antenna typically exhibits unequal azimuth and HPBWs in vertical transmissions, resulting in an elliptical disk in the ground plane. Even under ideal conditions where these two angles are identical, a slight tilt of the antenna from the vertical plane introduces an elliptical coverage area. In addition, inherently one-dimensional (1D) railways or highways benefit from elliptical patterns, which align more suitably with the elongated shape of the service area. This choice ensures more efficient coverage, directing communication signals specifically along the highway/railway.  Hence, modeling the user position uncertainty within elliptical footprints becomes essential to evaluate realistic link metrics such as distance distributions, outage probabilities, and throughput. Considering elliptical geometry, we offer a more robust and accurate framework that reflects the operational behavior of airborne communication platforms.   As a result, it is imperative to comprehensively assess the consequences of randomly placing users within this elliptical region to understand its impact on overall performance.

 {This paper is one of the first to systematically explore this issue, presenting analytical results that go beyond traditional circular models and addressing a literature gap in realistic air-to-ground footprint modeling. Unlike prior works that assume circular coverage for simplicity, this study addresses the analytical complexity introduced by elliptical footprints, which lack radial symmetry. This asymmetry complicates the derivation of distance distributions and requires advanced mathematical treatment. We derive closed-form expressions for the probability density function (PDF) and cumulative distribution function (CDF) of both radial and Euclidean distances between a UAV and a uniformly distributed ground user within elliptical coverage areas. These results enable for an exact evaluation of key metrics such as signal-to-noise ratio (SNR) statistics, outage probability, and multiuser (MU) throughput under Nakagami-$m$ fading. Our analysis highlights the impact of footprint shape, antenna tilt, and beamwidth asymmetry on coverage reliability and spectral efficiency, offering a rigorous basis for designing energy-efficient adaptive UAV communication systems with realistic antenna models.

Table~\ref{Table1} highlights the key distinctions between our work and several representative studies in the field. In contrast to prior research that primarily assumes circular footprints or relies on simulation-based evaluations, this work models realistic elliptical coverage resulting from directional antennas with tilt or beamwidth asymmetry. 

\begin{table*}[h]
\caption{Comparison with representative related works}
 
\label{Table1}
\centering
\renewcommand{\arraystretch}{1.2}
\footnotesize
\setlength{\tabcolsep}{4pt} 
\begin{tabular}{>{\raggedright\arraybackslash}p{2.5cm}|
                >{\raggedright\arraybackslash}p{2.5cm}|
                >{\raggedright\arraybackslash}p{2.5cm}|
                >{\raggedright\arraybackslash}p{2.5cm}|
                >{\raggedright\arraybackslash}p{2.5cm}|
                >{\raggedright\arraybackslash}p{2.5cm}}
\hline \hline
\textbf{Reference} & \textbf{Coverage Shape} & \textbf{User Distribution} & \textbf{Tx/Antenna Model} & \textbf{Key Metrics} & \textbf{Analytical Results} \\
\hline
Vaiopoulos \textit{et al.} \cite{J:Vaiopoulos3} & Circular & Uniform & Directional UAV & Outage Probability & Closed-Form Analytical Expression \\ 
Nafees \textit{et al.} \cite{J:Nafees21} & Circle Packing & Uniform & Directional UAV & Coverage Probability, Area Spectral Efficiency (ASE), Rate Coverage & Simulation-Based \\ 
Azari \textit{et al.} \cite{J:Azari} & Circular (Tilted) & Uniform & Directional UAV with Tilt & Coverage Probability, Throughput & Approximate Expressions \\ 
Tamo \textit{et al.} \cite{C:Tamo} & Elliptical & Uniform & Satellite Footprint & Point Generation Only & No Performance Metrics \\ 
\textbf{This Work} & \textbf{Elliptical (Tilted or HPBW-Asymmetric)} & \textbf{Uniform within Ellipse} & \textbf{Directional UAV (Tilted or Unequal HPBWs)} & \textbf{Distance Distributions, SNR, Outage, Throughput} & \textbf{Analytically Tractable Expressions} \\ 
\hline \hline
\end{tabular}
\normalsize
\end{table*}

The above observations are not restricted solely to communications between UAVs and ground users; they apply to any scenario involving an airborne Tx, including satellite communications, high-altitude platform (HAP) communications, HSR communications, and similar contexts. Moreover, exploring applications in VLC or underwater optical links presents exciting prospects. The present work does not comprehensively address all these cases; instead, it serves as an initial effort to elucidate the dynamics of a research area that is still inadequately explored in the existing technical literature.

\subsection{Contribution}
  Although several application studies are mentioned, such as UAV communications, satellite links, and high-speed railways, this paper specifically focuses on UAV-based air-to-ground communications. Consequently, the wireless channel is modeled using the Nakagami-$m$ fading distribution, which offers a flexible and widely accepted representation of small-scale fading in UAV environments.   We investigate two scenarios where the intersection with the ground creates elliptical footprints and randomly placed users are assumed. In summary, the current study provides the following significant advances:

\begin{figure*}[t]
\centering
\subfigure[\scriptsize $x_{0}=0$
]{\includegraphics[width=2in]{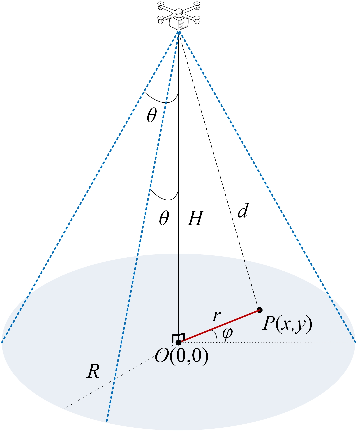}}
\hspace{0.1in} 
\subfigure[\scriptsize $0 < x_{0} \leq a$
]{\includegraphics[width=2.6in]{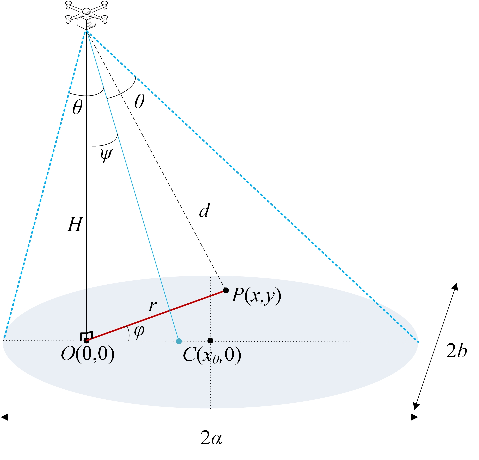}} 
\hspace{0in} 
\subfigure[\scriptsize $x_{0} > a$
]{\includegraphics[width=3.2in]{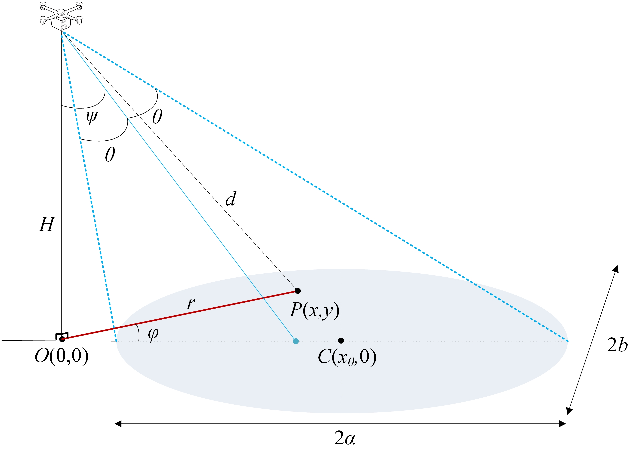}} 
\caption{Possible arrangements for (a) vertical transmission and (b), (c) tilted transmission.}
\label{Figure1}
\end{figure*}

\begin{itemize}
\item \textbf{First scenario: Random Rx location for tilted Tx radiation}: We consider an airborne Tx with an ideal directional antenna radiating downwards at a specified angle to form an elliptical footprint on the ground where a randomly placed user exists. We derive novel closed-form expressions for the PDF and CDF of the random distance between the Tx and the Rx positions in terms of simple and easy-to-evaluate functions.
\item \textbf{Second scenario: Random Rx location for vertical Tx radiation}: We repeat the analysis for an airborne Tx with a practical directional antenna radiating vertically downward to the ground. The expressions provided for the distance metrics offer insight into the statistical properties of the Tx-Rx distance and can be used in the design and analysis of any wireless system with similar characteristics.
\item \textbf{Application in UAV communications}: as a case study in both scenarios, we deduce the SNR statistics and evaluate the outage performance of a UAV-to-ground communication link under the assumption of Nakagami-$m$ fading. We compare the two scenarios and provide a series of figures in a clear and organized manner to highlight the results. An MU scenario is considered, and performance is evaluated against the number of users. The findings presented provide valuable information on the performance of similar wireless systems in fading environments.
\end{itemize}

\subsection{Organization}

The remaining work is structured as follows: Section II presents a concise model for the first scenario, assuming a Tx with a tilted directional antenna with equal azimuth and elevation HPBWs. This section elaborates on the derivations of critical statistical metrics, such as the joint PDF, radial distance, and Euclidean distance distributions for users within the ellipse. Section III explores the second scenario, which involves a Tx with a vertically transmitting directional antenna with unequal azimuth and elevation HPBWs. The same analysis presented in Section II is pursued and valuable outcomes are extracted. Furthermore, assuming Nakagami-$m$ fading, Section IV computes the SNR statistics and assesses outage performance. The two scenarios are compared, and representative analytical and numerical results are exhibited through a series of figures. In addition, a scenario for MU is investigated, and a throughput computation is performed. Finally, Section V summarizes the main findings of the study, discusses their implications, and suggests potential areas for future research in this field.

\textit{Notation}: All random variables will be denoted in bold plain font, e.g., $\mathbf{x}$, while the values of these variables will be in italic font, e.g., $x$. This notation also applies to the PDF and the CDF, e.g., $f_{\mathbf{x}}(x)$, $F_{\mathbf{x}}(x)$.

\subsection{Nomenclature}
In this paper a comprehensive enumeration of the primary symbols used in this study is presented to enhance reader convenience.
\nomenclature{\(a\)}{Major semi-axis of ellipse}
\nomenclature{\(b\)}{Minor semi-axis of ellipse}
\nomenclature{\(x_{0}\)}{Coordinate of the intersection of the ellipse's axes on $x$-axis}
\nomenclature{\(S\)}{Surface of ellipse}
\nomenclature{\(\mathcal{E}\)}{Eccentricity of ellipse}
\nomenclature{\(\theta\)}{Semi-apex angle of UAV's antenna}
\nomenclature{\(\psi\)}{Tilting angle of UAV's antenna}
\nomenclature{\(H\)}{Altitude of UAV}
\nomenclature{\((x,y)\)}{Cartesian coordinates of user position}
\nomenclature{\((r,\varphi)\)}{Polar coordinates of user position}
\nomenclature{\(\nu\)}{Path loss exponent}
\nomenclature{\(m\)}{Parameter of the Nakagami-$m$ distribution}
\nomenclature{\(R\)}{Radius of circular disc}
\nomenclature{\(d\)}{Euclidean distance}
\nomenclature{\(p\)}{Signal power}
\nomenclature{\(\overline{p}\)}{Mean signal power}
\nomenclature{\(P_{t}\)}{Transmitted power}
\nomenclature{\(P_{n}\)}{Noise power}
\nomenclature{\(\gamma\)}{Instantaneous SNR}
\nomenclature{\(\overline{\gamma}\)}{SNR at the Tx}
\nomenclature{\(M\)}{Number of users}
\nomenclature{\(P_{\rm{out}}\)}{Outage probability}
\nomenclature{\(\gamma_{\rm{th}}\)}{SNR threshold}
\nomenclature{\(\overline{C}_{MU}\)}{Throughput}

\printnomenclature

\section{First scenario: Tilted Transmission}
\subsection{Model Description}

Without loss of generality, we consider the wireless link between a hovering UAV and a stationary user with a randomly distributed position, denoted by $P$, within an ellipse, as illustrated in Fig. \ref{Figure1}. The UAV is positioned at a fixed altitude $H$, and is equipped with a directional antenna characterized by HPBWs of equal azimuth and elevation and a constant antenna gain, $G$ \cite{J:He2017}. This configuration results in the formation of a circular disc in the ground plane. This antenna can be tilted to $\psi$ relative to the line perpendicular to the terrain. The midpoint of both axes, $C$, lies on the $x$-axis with Cartesian coordinates $\{x_{0},0\}$, while the coordinates of $P$ are denoted by $\{x,y\}$. The projection of the UAV on the ground is the reference point indicated by $O$. The semi-apex angle $\theta$ and the other geometric parameters, including the radial distance $r$, the Euclidean distance $d$, the polar angle $\varphi$, the major axis 2$a$, and the minor axis 2$b$, are clearly represented. The surface area of an ellipse can be calculated as $S\triangleq\pi ab$ \cite[eq. (3.328a)]{B:Bronshtein}. Finally, the eccentricity of the ellipse, denoted by $\mathcal{E}$, is defined in terms of its axes as $\mathcal{E}\triangleq\sqrt{1-b^2/a^2}$ according to \cite[ch. 3.5.2.8.1]{B:Bronshtein}.

Three possible arrangements can be observed following the information demonstrated in Fig. \ref{Figure1}. The first originates when $\psi = 0^\circ$, resulting in a vertically radiating UAV and a circular footprint. The remaining two are related to the creation of elliptical footprints when the UAV changes its inclination based on whether $O$ lies within or outside the ellipse. In view of the above, we provide a series of derivations to calculate critical statistical metrics that depend on the relative position of $x_{0}$ on the positive semi-major axis of $x$.

\subsection{Dimensions of the Ellipse}
While maintaining a constant value for $2\theta$ and $H$, the dimensions of the ellipse are subject to variation as the value of $\psi$ increases. The dimensions of the major and minor axes of the ellipse can be estimated through the application of a set of mathematical equations, as detailed in \cite{J:Azari}
\begin{equation}
\begin{split}
a\triangleq\dfrac{H\sin(2\theta)}{2(\cos^{2}\psi-\sin^{2}\theta)}, \\
b\triangleq\dfrac{H\sin \theta}{\sqrt{\cos^{2}\psi-\sin^{2}\theta}}.
\end{split}
\label{EllDim}
\end{equation}
Furthermore, $x_{0}$ is determined as $x_{0}=a-H\tan(\theta-\psi)$ and $x_{0}=a+H\tan(\psi-\theta)$ for the arrangements illustrated in Fig. \ref{Figure1}(b) and \ref{Figure1}(c), respectively.

Adjustments in tilt angle directly impact both the coverage area and signal strength. A larger tilt angle increases the footprint size, but reduces the quality of the received signal at the edges due to longer distances and lower directional gain. A smaller tilt angle improves the SNR near the UAV, but limits coverage. This trade-off is critical in the design and deployment of UAV antennas. As discussed in \cite{J:Messaoudi}, optimizing spatial coverage under such constraints supports energy-efficient and delay-aware UAV operations. Our analysis quantifies these effects through outage probability, offering practical guidance for selecting tilt angle.

Figure~\ref{Figure2}(a) illustrates the expansion of the elliptical coverage region with an increasing tilt angle $\psi$, at a fixed UAV altitude of $H = 300$m. In practical UAV deployments, the antenna tilt angle $\psi$ is typically in the range of $0^\circ$ to $45^\circ$, depending on the desired displacement of the footprint and the requirements of ground coverage. The semi-apex angle $\theta$, which is related to the beamwidth of the directional antenna, commonly ranges from $20^\circ$ to $60^\circ$, consistent with realistic high-gain antenna patterns. The aspect ratio of the ellipse, defined as $a/b$, varies accordingly with both $\psi$ and $\theta$, but typically falls within the range of $1$ (circular footprint) to approximately $2.5$, as observed in directional or tilted transmission scenarios (see Figs.~\ref{Figure2}(b)-(d)). These ranges are in alignment with existing studies such as \cite{J:Azari} and \cite{J:He2017}, and are used throughout the numerical evaluations in this work.

For small tilt angles, for example, $\psi = 10^{\circ}$, the coverage area shows limited growth across all values of the semi-apex angle $\theta$. In contrast, for $\psi > 30^{\circ}$, the coverage area increases significantly with $\theta$. For example, at $\theta = 30^{\circ}$, the area grows by a factor of approximately 1.2 at $\psi = 20^{\circ}$, and by about 3.3 at $\psi = 40^{\circ}$. This trend is further visualized in Figs.~\ref{Figure2}(b)–(d), where the elliptical footprint becomes progressively larger and more elongated as $\psi$ increases. These results highlight the strong influence of the tilt angle $\psi$ on elliptical and circular coverage geometries.

\begin{figure}[h]
\centering
\subfigure[\scriptsize $H=300$m
]{\includegraphics[width=0.36\textwidth]{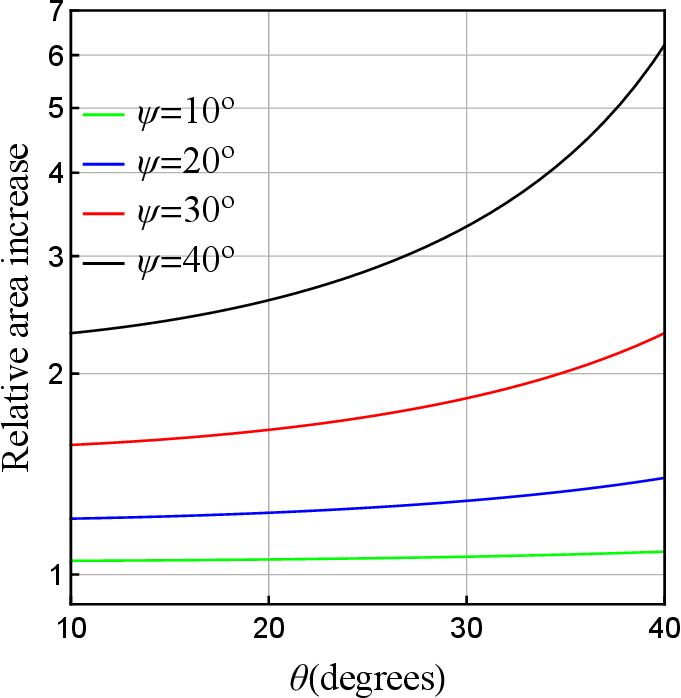}} 
\hspace{0in}
\subfigure[\scriptsize $H=300$m, $\theta=20^{\circ}$
]{\includegraphics[width=0.4\textwidth]{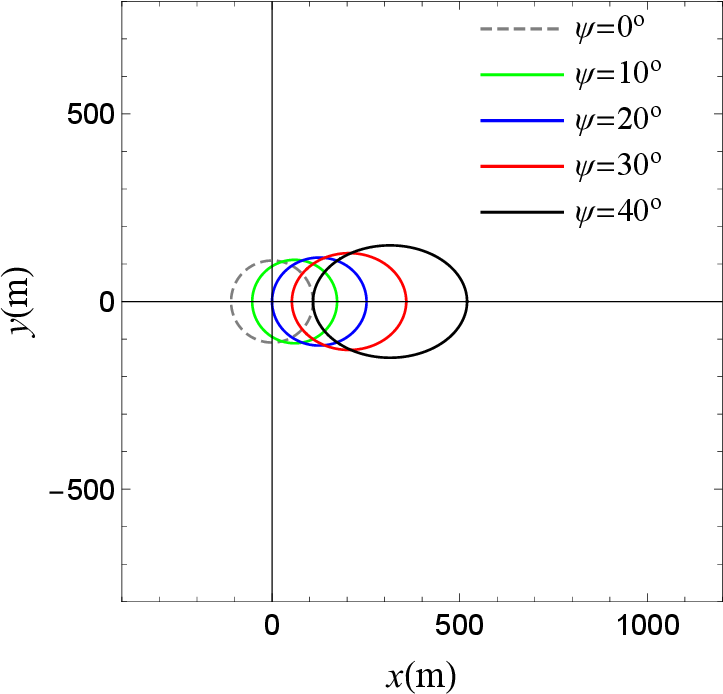}} 
\hspace{0in}\\
\subfigure[\scriptsize $H=300$m, $\theta=30^{\circ}$
]{\includegraphics[width=0.4\textwidth]{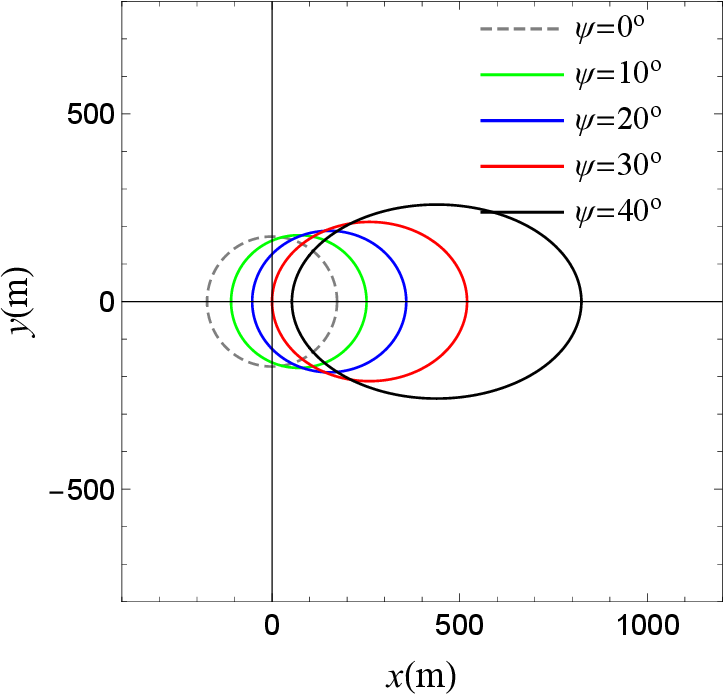}} 
\hspace{0in} 
\subfigure[\scriptsize $H=300$m, $\theta=40^{\circ}$
]{\includegraphics[width=0.4\textwidth]{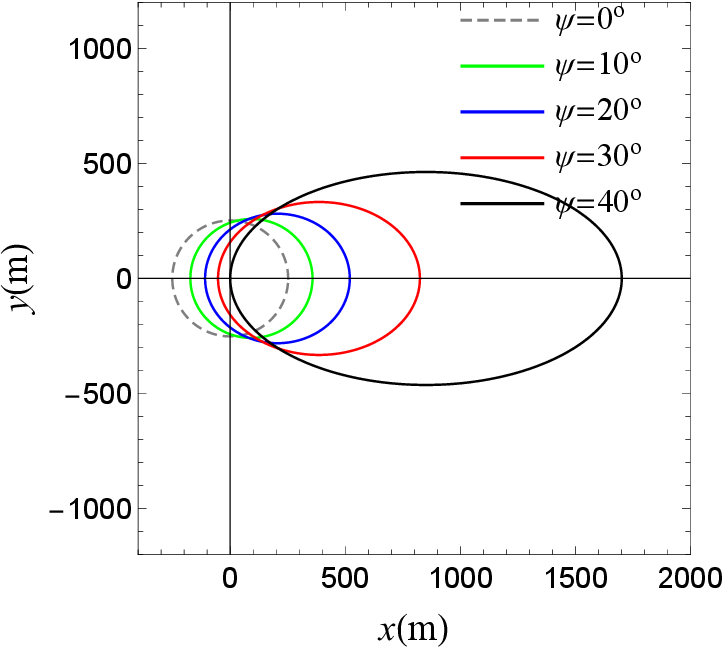}} 
\caption{(a) Relative increase of the elliptical area as a function of $\theta$. (b)–(d) Perimeter of the elliptical area versus $\psi$ for different values of $\theta$.}
\label{Figure2}
\end{figure}

\subsection{Radial Distance Distribution}
\label{Radial2}
The joint PDF of the uniform random location of $P$ inside the elliptical disk is given by
\begin{equation}
f_{\mathbf{x,y}}(x,y) = \left\{\hspace{-5pt}
\begin{array}{ll}
\dfrac{1}{S}=\dfrac{1}{\pi a b}, &
b^{2}x^{2}+a^{2}y^{2}\leq a^{2}b^{2} \\
& \\
0, & \rm{else},
\end{array}
\right. \label{pdfCart}
\end{equation}
Applying the two-dimensional random variable (RV) transformation method to polar coordinates given by $\mathbf{x} = \mathbf{r} \cos \bm{\varphi}, \mathbf{y}= \mathbf{r} \sin \bm{\varphi}$ \cite[eq. (6.115)]{B:Papoulis}, we arrive at
\begin{equation}
\begin{array}{ll}
f_{\mathbf{r},\bm{\varphi}}(r,\varphi)=\dfrac{r}{\pi ab}, & (r,\varphi)\in \left\{\mathcal{D}_{t}\right\}.
\end{array}
\label{jpdf1}
\end{equation}
where the set $\mathcal{D}_{t}$ is defined as the union of discrete subsets, corresponding to the inscribed circle of radius $b$ and the remaining area of the ellipse, respectively, as
\begin{itemize}
\item $0 < x_{0} \leq a$
\begin{eqnarray}
\mathcal{D}_{t} &\triangleq&\left\{ (r,\varphi)\ \middle|\ 
\begin{array}{c}
0\leq r\leq a-x_{0}, \\ 
0\leq \varphi \leq 2\pi
\end{array}
\right\} \cup   \notag \\
&&\left\{ (r,\varphi)\ \middle|\ 
\begin{array}{c}
a-x_{0} < r\leq a+x_{0}, \\ 
2\pi-\varphi(r) \leq \varphi \leq \varphi(r)
\end{array}
\right\},
\end{eqnarray}
\item $x_{0} > a$
\begin{eqnarray}
\mathcal{D}_{t} &\triangleq&\left\{ (r,\varphi)\ \middle|\ 
\begin{array}{c}
x_{0}-a \leq r\leq a+x_{0}, \\ 
2\pi-\varphi(r) \leq \varphi \leq \varphi(r)
\end{array}
\right\}
\end{eqnarray}
\end{itemize}
with the angle $\varphi(r)$ defined below in \eqref{phi1}.
Then, the radial distance distribution, $f_{\mathbf{r}}(r)$, is evaluated by the following theorem,
\begin{theorem}
Based on the joint PDF expression \eqref{jpdf1} and the relative position of $x_{0}$, the radial distance distribution, $f_{\mathbf{r}}(r)$, is extracted as

\item{$\bullet$} $0 < x_{0} \leq a$
\begin{equation}
\hspace{-3pt}f_{\mathbf{r}}(r)\hspace{-2pt}=\hspace{-2pt}\left\{ 
\begin{array}{l}
\hspace{-5pt}\dfrac{2r}{ab},\text{ \ }0\leq r\leq a-x_{0} \\ 
\\ 
\hspace{-5pt}\dfrac{2r\cos ^{-1}\left( \dfrac{k_{1}+\Lambda _{1}(a,b,r)}{\mathcal{E}%
^{2}a^{2}r}\right) }{\pi ab},\text{ }a-x_{0}<r\leq a+x_{0}, \\ 
\end{array}%
\right.   \label{pdfr1a}
\end{equation}

\item{$\bullet$} $x_{0} > a$
\begin{equation}
f_{\mathbf{r}}(r)=\dfrac{2r\cos^{-1}\left(\dfrac{k_{1}+\Lambda_{1}(a,b,r)}{\mathcal{E}^{2}a^{2}r}\right)}{\pi ab}, \hspace{10pt} x_{0}-a<r \leq a+x_{0},
\label{pdfr1b}
\end{equation}\\
where $k_{1}\triangleq-b^{2}x_{0}$ and $k_{2}\triangleq\sqrt{b^{2}(\mathcal{E}^{2}a^{2}-x_{0}^{2})}$ are system-dependent constants, while $\Lambda_{1}(a,b,x)\triangleq a\sqrt{\mathcal{E}^{2}a^{2}x^{2}-k_{2}^{2}}$, is a user-defined function\footnote{$\Lambda_{1}(a,b,x)$, lacks explicit physical interpretation but offers a more concise representation of the final expressions, e.g., (\ref{pdfr1a}) or (\ref{pdfr1b}).}.
\end{theorem}
\begin{proof}
The general equation of the ellipse in polar form is as \cite[ch. 3.5.2]{B:Bronshtein}
\begin{equation}
\label{polar}
\dfrac{(r\cos\varphi-x_{0})^2}{a^2}-\dfrac{(r\sin\varphi)^2}{b^2}=1.
\end{equation}
Solving \eqref{polar} for $\varphi$, we get
\begin{equation}
\label{phi1}
\varphi(r)=\cos^{-1}\left(\dfrac{k_{1} + \Lambda_1(a,b,r)}{\mathcal{E}^2a^2r} \right).
\end{equation}
The radial distance distribution is obtained by integrating \eqref{jpdf1} over $\varphi(r)$. 
\begin{figure}[h]
\centering
\subfigure[\scriptsize $0 < x_{0} \leq a$
]{\includegraphics[width=4in]{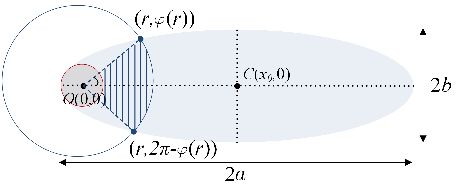}} 
\hspace{0.1in} 
\subfigure[\scriptsize $x_{0} > a$
]{\includegraphics[width=4in]{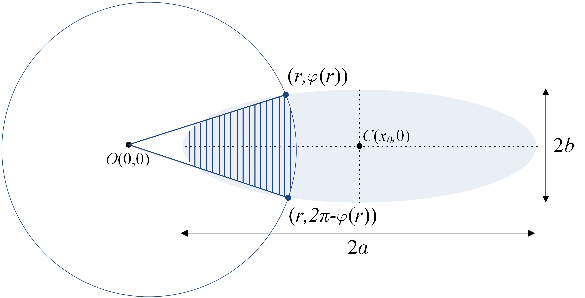}} 
\caption{Integration areas used in \eqref{jpdf1} to derive \eqref{pdfr1a} or \eqref{pdfr1b}.}
\label{Figure3}
\end{figure}
We differentiate between two cases based on the location of $x_{0}$ on the positive semi-axis $x$. This results in a varying number of intersections between the inscribed circle (centered at $(0,0)$) and the ellipse, as shown in Fig. \ref{Figure3}. In the first case, the inscribed circle intersects the ellipse at two points, and the integration over the polar angle is carried out in separate branches. The first branch applies where the inscribed circle exists and the integration is performed between 0 and 2$\pi$. In the second case, there is no intersection, and the integration is defined between the angles obtained from \eqref{phi1}.
\end{proof}

\subsection{Euclidean Distance Statistics}
In the sequel, the statistical description of the RV $\mathbf{d}$ becomes as follows.
\begin{theorem}
The PDF of the Euclidean distance $\mathbf{d}$  is given by

\item {$\bullet $} $0<x_{0}\leq a$ 
\begin{equation}
\hspace{-5pt}f_{\mathbf{d} }(d)=\left\{ 
\begin{array}{l}
\dfrac{2d}{ab},\text{ \ \ \  }d_{\mathrm{min}}\leq d\leq d_{(1),t} \\ 
\\ 
\dfrac{2d\cos ^{-1}\left( \dfrac{k_{1}+\Lambda _{1}(a,b,\sqrt{d^{2}-H^{2}})}{%
\mathcal{E}^{2}a^{2}\sqrt{d^{2}-H^{2}}}\right) }{\pi ab},\text{ } \\ 
\text{ \  }d_{(1),t}<d\leq d_{\mathrm{{max},t}}%
\end{array}%
\right.   \label{pdfd1a}
\end{equation}

\item {$\bullet $} $x_{0}>a$ 
\begin{eqnarray}
\hspace{-5pt}f_{\mathbf{d} }(d) \hspace{-5pt}&=&\hspace{-5pt}\dfrac{2d\cos ^{-1}\left( \dfrac{k_{1}+\Lambda _{1}(a,b,\sqrt{%
d^{2}-H^{2}})}{\mathcal{E}^{2}a^{2}\sqrt{d^{2}-H^{2}}}\right) }{\pi ab}, \notag \\
\hspace{10pt} &&d_{(2),t}<d\leq d_{\mathrm{{max},t}} \label{pdfd1b}
\end{eqnarray}%
where $d_{\mathrm{min}}\triangleq H$, $d_{(1),t}\triangleq \sqrt{%
(a-x_{0})^{2}+H^{2}}$, $d_{(2),t}\triangleq \sqrt{(x_{0}-a)^{2}+H^{2}}$, $d_{%
\mathrm{{max},t}}\triangleq \sqrt{(a+x_{0})^{2}+H^{2}}$ are system-dependent
constants.
\end{theorem}

\begin{proof}
The Euclidean distance, $d$, is calculated using the Pythagorean theorem $d=\sqrt{r^2+H^2}$. 
The PDF of $\mathbf{d} $ results from \eqref{pdfr1a} and \eqref{pdfr1b} by applying the RV transformation method described in \cite[eq. (5.16)]{B:Papoulis} as
\begin{equation}
f_{\mathbf{d} }(d)=\dfrac{d\cdot f_{\mathbf{r} }(\sqrt{d^2-H^2})}{\sqrt{d^2-H^2}}.
\label{varchd}
\end{equation}
\end{proof}

\begin{figure}[h]
\centering
\subfigure[\scriptsize $\psi=10^{\circ}$
]{\includegraphics[width=0.4\textwidth]{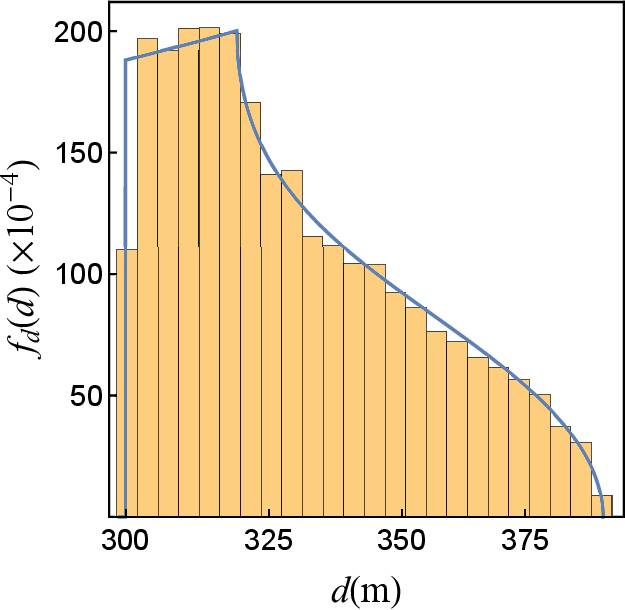}} 
\hspace{0in} 
\subfigure[\scriptsize $\psi=20^{\circ}$
]{\includegraphics[width=0.4\textwidth]{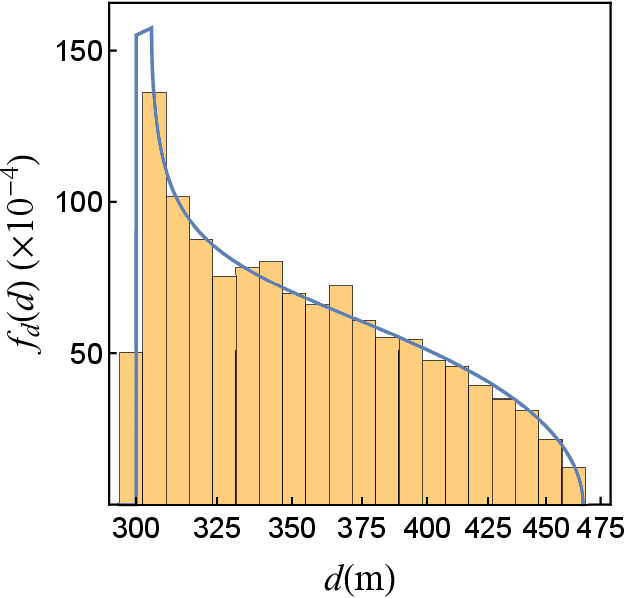}} 
\hspace{0in} 
\subfigure[\scriptsize $\psi=30^{\circ}$
]{\includegraphics[width=0.4\textwidth]{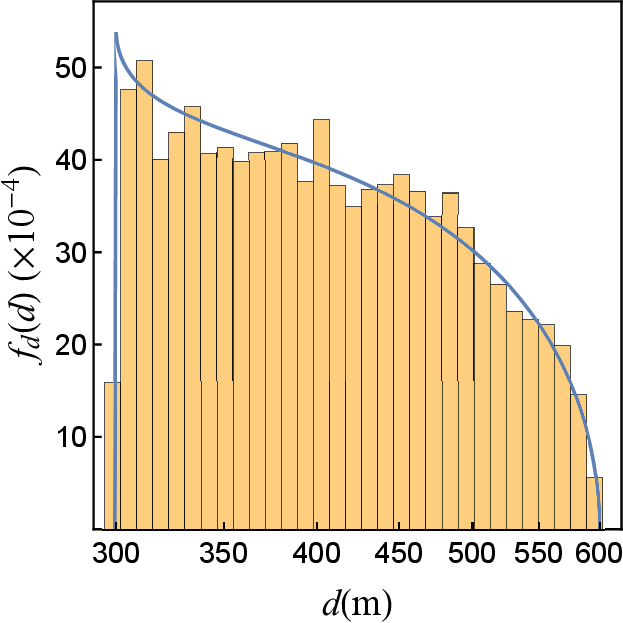}} 
\hspace{0in} 
\subfigure[\scriptsize $\psi=40^{\circ}$
]{\includegraphics[width=0.4\textwidth]{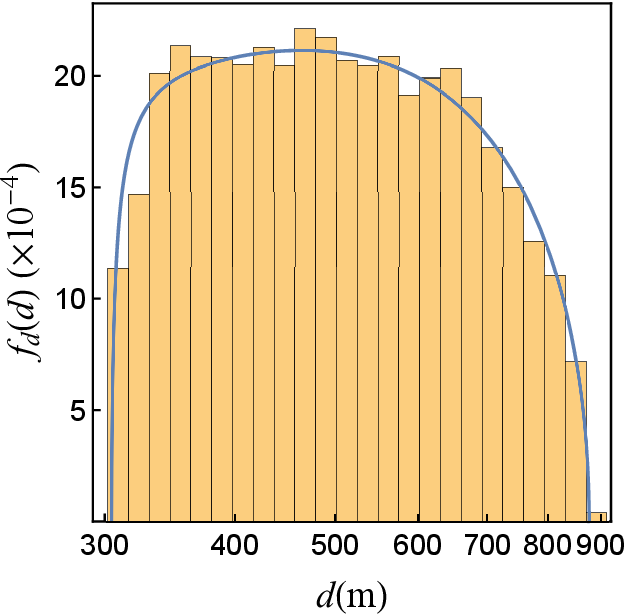}} 
\caption{PDF of $\mathbf{d}$ for typical $\psi$ values ($H=300$m and $\theta=30^{\circ}$): Analytical solutions (solid line) vs. Monte-Carlo simulation (histogram).}
\label{Figure4}
\end{figure}

Figure \ref{Figure4} offers a graphical representation of the PDF of $\mathbf{d} $ for different $\psi$ values, assuming $H=300$m and $\theta=30^\circ$. The solid lines are based on theoretical calculations (according to \eqref{pdfd1a}, \eqref{pdfd1b}) and have been verified through simulation results (histogram). These simulations validate the correctness of the analytical curves, giving confidence in the accuracy of the theoretical results.

\begin{corollary}
The CDF of the distance $\mathbf{d} $ is given by

\item{$\bullet$} $0 < x_{0} \leq a$
\begin{equation}
F_{\mathbf{d} }(d)=\left\{ 
\begin{array}{ll}
\dfrac{d^{2}-H^{2}}{ab}, & d_{\mathrm{min}}\leq d\leq d_{(1),t} \\ 
&  \\ 
\begin{array}{l}
\dfrac{1}{2}+\dfrac{d^{2}-H^{2}}{ab} \\ 
+\Lambda _{2}(a,b,\sqrt{d^{2}-H_{2}}),%
\end{array}
& d_{(1),t}<d\leq d_{\mathrm{{max},t}} \\ 
& 
\end{array}%
\right.   \label{cdfd1a}
\end{equation}

\item{$\bullet$} $x_{0} > a$
\begin{eqnarray}
F_{\mathbf{d} }(d) &=&\dfrac{1}{2}+\dfrac{d^{2}-H^{2}}{ab}+\Lambda _{2}(a,b,\sqrt{%
d^{2}-H_{2}}),\hspace{10pt} \notag \\
&&d_{(2),t}<d\leq d_{\mathrm{{max},t}}  \label{cdfd1b}
\end{eqnarray}%
where
\begin{align*}
\Lambda _{2}(a,b,x)\triangleq & \dfrac{x^{2}\cos ^{-1}\left( \frac{%
k_{1}+\Lambda _{1}(a,b,x)}{\mathcal{E}^{2}a^{2}x}\right) }{\pi ab} \\
& -\dfrac{k_{5}+k_{6}x^{2}-2k_{7}\Lambda _{1}(a,b,x)+}{\pi abk_{10}x\sqrt{1+%
\frac{\left( k_{1}-\Lambda _{1}(a,b,x)\right) ^{2}}{\mathcal{E}^{2}a^{2}x}}}
\\
& -\dfrac{k_{8}\Lambda _{3}(a,b,x)\tan ^{-1}\left( \frac{-k_{1}\mathcal{E}%
^{2}a^{4}-k_{4}\Lambda _{1}(a,b,x)}{k_{9}\Lambda _{3}(a,b,x)}\right) }{\pi
abk_{10}x\sqrt{1+\frac{\left( k_{1}-\Lambda _{1}(a,b,x)\right) ^{2}}{%
\mathcal{E}^{2}a^{2}x}}},
\end{align*}%
$\Lambda_{3}(a,b,x)\triangleq\sqrt{k_{3}+k_{4}x^{2}-2k_{1}\Lambda_{1}(a,b,x)}$
are user-defined functions, while $k_{3}\triangleq ak_{2}^{2}-k_{1}^{2}$, $k_{4}\triangleq \mathcal{E}^{2}a^{4}(1-{\mathcal E}^2)$, $k_{5}\triangleq -a^{6}\mathcal{E}^{3}(1-\mathcal E^2)k_{1}k_{3}$, $k_{6}\triangleq -k_{1}a^{12}\mathcal{E}^{5}(1-\mathcal E^2)^{2}$, $k_{7}\triangleq a^{6}\mathcal{E}^{3}(1-\mathcal E^2)k_{1}$, $k_{8}\triangleq -a^{6}\mathcal{E}^{3}\sqrt{1-{\mathcal E}^2}((1-\mathcal E^2)a^{2}k_{2}^{2}+k_{1}^{2})$, $k_{9}\triangleq a^{4}\mathcal{E}^{2}\sqrt{1-{\mathcal E}^2}$, and $k_{10}\triangleq -a^{12}\mathcal{E}^{7}(1-\mathcal E^2)^{2}$ are system-dependent constants.
\end{corollary}

\begin{proof}
The CDF results by integrating \eqref{pdfd1a} and \eqref{pdfd1b} in terms of $d$.
\end{proof}

\section{Second scenario: Vertical Transmission}
\subsection{Model Description}

We consider the configuration in Fig. \ref{Figure5} illustrating an air-to-ground link between a Tx with a directional antenna and a uniformly distributed user at a point, $P$ with Cartesian coordinates $\{x,y\}$, within the elliptical area and a Tx at a certain altitude. The HPBWs of the azimuth and elevation antennas are not equal, thus forming an elliptical coverage area \cite{J:Oubbati}, \cite{J:Jeon}. The reference point for this area is identified as the midpoint of the two axes, denoted by $O$, which coincides with the projection of the UAV on the ground. The remaining geometric parameters -- including altitude $H$, radial distance $r$, Euclidean distance $d$, polar angle $\varphi$, major axis 2$a$, and minor axis 2$b$ -- are adequately highlighted.
\begin{figure}[h]
\centering
\includegraphics[keepaspectratio,width=4in]{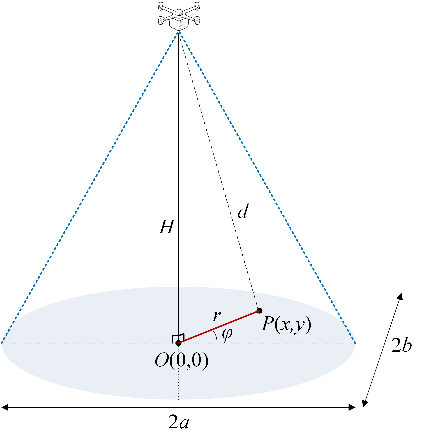}
\caption{Model configuration.}
\label{Figure5}
\end{figure}

\subsection{Radial Distance Distribution}
The joint PDF of the uniform random location of $P$ inside the elliptical disk is obtained as previously.

\begin{equation}
\begin{array}{ll}
f_{\mathbf{r},\bm{\varphi} }(r,\varphi)=\dfrac{r}{\pi ab}, & (r,\varphi)\in \left\{\mathcal{D}_{\rm{v}}\right\},
\end{array}
\label{jpdf2}
\end{equation}
where the set $\mathcal{D}_{\rm{v}}$ is defined as the union of three discrete subsets, corresponding to the circle of radius inscribed $b$ and the remaining area of the ellipse, respectively.
\begin{eqnarray}
\mathcal{D}_{\rm{v}} &\triangleq&\left\{ (r,\varphi)\ \middle|\ 
\begin{array}{c}
0\leq r\leq b, \\ 
0\leq \varphi \leq 2\pi
\end{array}
\right\} \cup   \notag \\
&&\left\{ (r,\varphi)\ \middle|\ 
\begin{array}{c}
b < r\leq a, \\ 
2\pi-\varphi(r) \leq \varphi \leq \varphi(r)
\end{array}
\right\} \cup
\notag \\
&&\left\{ (r,\varphi)\ \middle|\ 
\begin{array}{c}
b < r\leq a, \\ 
\pi-\varphi(r) \leq \varphi \leq \pi+\varphi(r)
\end{array}
\right\},
\end{eqnarray}
with the angle $\varphi(r)$ defined below in \eqref{phi2}. The radial distance distribution is extracted according to the following theorem.

\begin{theorem}
Based on the joint PDF expression \eqref{jpdf2} the radial distance distribution, $f_{\mathbf{r}}(r)$, is extracted as

\begin{equation}
f_{\mathbf{r}}(r)=\left\{
\begin{array}{ll}
\dfrac{2r}{ab}, & 0\leq r \leq b \\
& \\
\dfrac{4r\cos^{-1}\left(\dfrac{\sqrt{r^{2}-b^{2}}}{\mathcal{E}r}\right)}{\pi ab}, & b<r \leq a. \\
\end{array}
\right. \label{pdfr2}
\end{equation}
\end{theorem}
\begin{proof}
Setting $x_{0}=0$ in \eqref{polar} and solving for $\varphi$, we get
\begin{equation}
\label{phi2}
\varphi(r)=\cos^{-1}\left(\dfrac{\sqrt{r^2-b^2}}{\mathcal{E} r} \right).
\end{equation}
The distribution of radial distance occurs by integrating \eqref{jpdf2} over $\varphi(r)$. A snapshot is displayed in Fig. \ref{Figure6}. It should be clarified that the integration of \eqref{jpdf2} is branch-wise. The first branch holds for the inscribed circle, where the integration over the polar angle takes place between $0$ and 2$\pi$.  In the other case, the integration areas, depicted as the shaded regions in Fig. \ref{Figure6}, are defined between the highlighted angles, which are determined on the basis of \eqref{phi2}. 

\begin{figure}[h]
\centering
\includegraphics[keepaspectratio,width=4in]{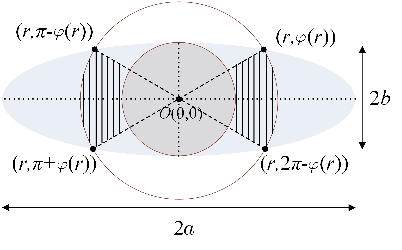}
\caption{Integration areas used in \eqref{jpdf2} to derive \eqref{pdfr2}.}
\label{Figure6}
\end{figure}
\end{proof}

It should be noted that as the ellipsoidal footprint approaches a circular one, \eqref{pdfr2} converges to the well-established expression of $f_{\mathbf{r}}(r) = 2r/R^2$, with $R$ being the circular disc radius. 

\subsection{Euclidean Distance Statistics}

\begin{theorem}
The PDF of the Euclidean distance $\mathbf{d}$ is given by
\begin{equation}
f_{\mathbf{d}}(d)\hspace{-3pt}=\hspace{-3pt}\left\{
\begin{array}{ll}
\hspace{-5pt}\dfrac{2 d}{a b}, & \hspace{-7pt}d_{\rm{min}}\leq d\leq d_{(1),{\rm{v}}} \\
& \\
\hspace{-5pt}\dfrac{4 d\cos ^{-1}\left(\sqrt{\dfrac{d^2-d_{(1),{\rm{v}}}^2}{\mathcal{E}^2   \left(d^2-H^2\right)}}\right)}{\pi  a b}, & \hspace{-7pt}d_{(1),{\rm{v}}}< d\leq d_{\rm{max},v}. \\
\end{array}
\right. \label{pdfd2}
\end{equation}
where $d_{(1),{\rm{v}}}\triangleq \sqrt{b^2+H^2}$ and $d_{\rm{max},v}\triangleq \sqrt{a^2+H^2}$ are system-dependent constants.
\end{theorem}
\begin{proof}
Eq. \eqref{pdfd2} yields after applying \eqref{varchd}.
\end{proof}

Figure \ref{Figure7} depicts the PDF of $\mathbf{d}$, assuming $a=180$m, $H=300$m, and two values of the ratio $a/b$. Figure \ref{Figure7} (a) considers $a/b=2$, while Fig. \ref{Figure7} (b) refers to the scenario where the ellipse degenerates into a circle as a limiting case. The validity is also confirmed through simulation.

\begin{corollary}
The CDF of $\mathbf{d}$ is given by 
\begin{equation}
{\small F_{\mathbf{d}}(d)=\left\{ 
\begin{array}{l}
\dfrac{d^{2}-H^{2}}{ab},\text{ \ }d_{\mathrm{{min}}}\leq d\leq d_{(1),%
\mathrm{v}} \\ 
\\ 
\begin{array}{l}
\dfrac{2(d^{2}\hspace{-3pt}-\hspace{-3pt}H^{2})\sqrt{1-\mathcal{E}^{2}}\cos
^{-1}\hspace{-3pt}\left( \hspace{-3pt}\dfrac{\sqrt{d^{2}-d_{(1),\mathrm{v}%
}^{2}}}{\mathcal{E}\sqrt{d^{2}-H^{2}}}\hspace{-3pt}\right) }{\pi ab\sqrt{1-%
\mathcal{E}^{2}}} \\ 
+\dfrac{2b^{2}\sin ^{-1}\hspace{-3pt}\left( \hspace{-3pt}\dfrac{\sqrt{(1-%
\mathcal{E}^{2})(d^{2}-d_{(1),\mathrm{v}}^{2}})}{b\mathcal{E}}\hspace{-3pt}%
\right) }{\pi ab\sqrt{1-\mathcal{E}^{2}}},%
\end{array}
\\ 
d_{(1),\mathrm{v}}<d\leq d_{\mathrm{{max},v}}. \\ 
\end{array}%
\right. }  \label{cdfd2}
\end{equation}
\end{corollary}
\begin{proof}
The CDF arises by integrating \eqref{cdfd2} in terms of $d$.
\end{proof}
\begin{figure}[h]
\centering
\subfigure[\scriptsize $a$=180m for $a/b$=2, and $H$=300m
]{\includegraphics[width=0.4\textwidth]{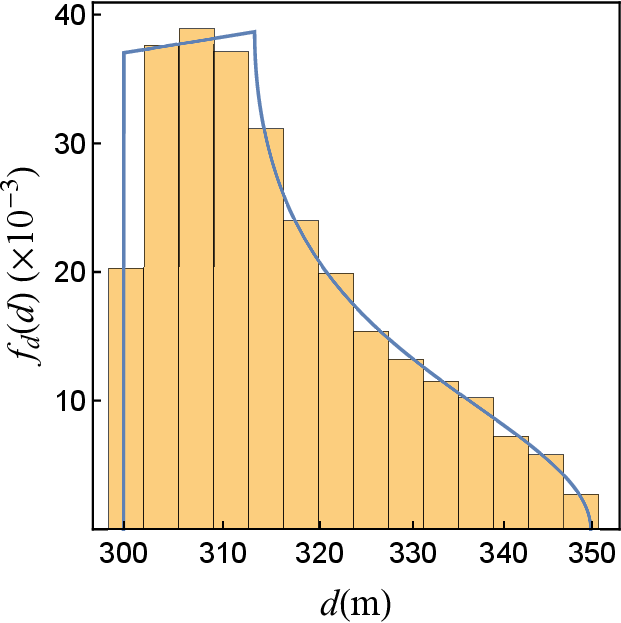}} 
\hspace{0in} 
\subfigure[\scriptsize $a$=180m, $a/b$=1, and $H$=300m
]{\includegraphics[width=0.4\textwidth]{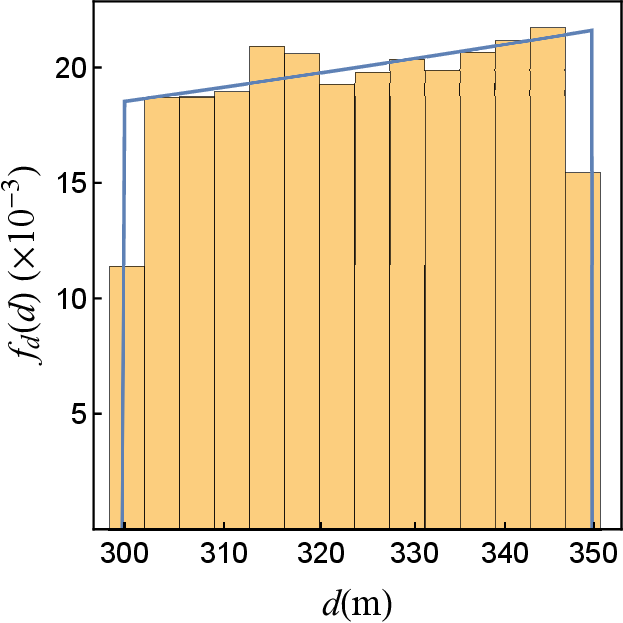}} 
\caption{PDF of $\mathbf{d}$. Analytical solutions (solid line) vs. Monte-Carlo simulation (histogram)}
\label{Figure7}
\end{figure}
\section{Performance Analysis}
  In this section, we evaluate the impact of antenna configuration, specifically the tilt angle $\psi$ and the semi-apex angle $\theta$, on system performance through the resulting elliptical footprint. Primary evaluation criteria include the statistical characterization of the SNR, the probability of outage (i.e., the probability that the received SNR falls below a predefined threshold) and the MU throughput. The analysis considers both single-user and MU cases under Nakagami-$m$ fading and accounts for how geometric parameters influence the distribution of users within the coverage area and their link quality to the UAV Tx. 

The circular models in \cite{J:Azari, J:Lai20} simplify the analysis, but overlook the antenna tilt asymmetry or the imbalance of beamwidth. The hexagonal and strip-shaped footprints in \cite{J:Lin} suit fixed layouts but lack flexibility for UAV mobility. Our model captures realistic elliptical coverage under directional antennas and Nakagami-$m$ fading. The results show improved outage and SNR performance, offering a more accurate and adaptable framework for UAV-based systems.

\subsection{SNR Statistics}
A thorough analysis of the SNR distribution is carried out in the context of a fading channel scenario, where the received signal includes both diffuse and specular components. In this scenario, the signal power, $\mathbf{p}$, adheres to a Nakagami-$m$ distribution, as described in \cite[eq. (3.39)]{B:Goldsmith}. The Nakagami-$m$ fading model, extensively employed in wireless communication systems, accurately characterizes the small-scale fading of the air-ground channel based on findings in \cite{C:Yanmaz} and \cite{C:Khawaja2016}. Unlike the Rayleigh fading model, the Nakagami-$m$ model introduces the parameter "$m$" to indicate the severity of the fading, thus accommodating a broader range of fading scenarios, including severe and mild conditions \cite{B:Singhal}.   Specifically, when $m = 1$, the Nakagami-$m$ distribution is exactly reduced to the Rayleigh distribution, which models environments without a line of sight (LOS) component~\cite{B:Simon}. For $m > 1$, it closely approximates the Rician distribution - commonly used for channels with a dominant LOS path and scattered multipath components - through a one-to-one correspondence between the Nakagami-$m$ parameter and the Rician factor $K$ \cite[eq. (2.26)]{B:Simon}. In the limiting case, as $m \to \infty$, the channel converges to an additive white gaussian noise (AWGN) model, representing a non-fading scenario.   This flexibility enables a more precise representation of the variability recognized in real conditions. That is, 

\begin{equation}
f_{\mathbf{p}}(p)=\left( \dfrac{m}{\overline{p}}\right) ^{m}\dfrac{p^{m-1}}{\Gamma (m)}
\mathrm{e}^{-\frac{mp}{\overline{p}}}, \label{Nakagami}
\end{equation}
with $\overline{p}\triangleq P_{t}/d^{\nu }$ being the mean signal power, $P_{t}$ the
transmitted power, $m\geq 1/2$ the Nakagami-$m$ fading parameter, $\nu$
the path loss exponent, and $\Gamma(\cdot)$ the Gamma function\cite[eq. (06.05.02.0001.01)]{Wolfram}. 
The conditional PDF of the received signal power with
respect to distance $\mathbf{d}$ is determined and averaged over $\mathbf{d}$ as
\begin{eqnarray}
f_{\mathbf{p|d}}(p|d)&=&\left(\frac{md^{\nu }}{P_{t}}\right) ^{m}\dfrac{ p^{m-1}}{\Gamma (m)}\mathrm{e}^{-\frac{mpd^{\nu }}{P_{t}}}, \label{CondPower}
\end{eqnarray}
Then, the unconditional power distribution can be deduced after integration over $\mathbf{d}$ in each case.
\begin{equation}
f_{\mathbf{p}}(p)=\int f_{\mathbf{p|d}}(p|d)f_{\mathbf{d}}(d)d\rm{d}, \label{UnPower}
\end{equation}
The integral in \eqref{UnPower} is finite, with limits determined by \eqref{pdfd1a} and \eqref{pdfd1b}, or by \eqref{pdfd2}, respectively.

\subsubsection{First scenario}
The PDF of the instantaneous SNR $\textbf{\textgamma}\triangleq \frac{\mathbf{p}}{P_n}$, where $P_{n}$ is the noise power, occurs according to the following RV transformation \cite{J:Vaiopoulos}
\begin{equation} \label{pdfSNR}
f_{\textbf{\textgamma}}(\gamma)=P_{n}f_{\mathbf{p}}(\gamma P_n).
\end{equation}
 where $f_{\mathbf{p}}(p)$ is the unconditional power distribution according to \eqref{UnPower}.
\begin{itemize}
\item[$\bullet$] $0 < x_{0} \leq a$
\begin{eqnarray}
\hspace{-15pt}f_{\textbf{\textgamma}}(\gamma ) \hspace{-3pt}&=&\hspace{-3pt}\frac{2\overline{\gamma }^{\frac{2}{\nu }}\gamma ^{-%
\frac{\nu +2}{\nu }}}{ab\nu \Gamma (m)m^{\frac{2}{\nu }}}\left( \Gamma
\left( m+\frac{2}{\nu },\frac{\gamma md_{\mathrm{min}}^{\nu }}{\overline{%
\gamma }}\right) \right.   \notag \\
&&-\left. \Gamma \left( m+\frac{2}{\nu },\frac{\gamma md_{(1),t}^{\nu }}{%
\overline{\gamma }}\right) \right) +\mathcal{I}_{t}(\gamma ),  \label{pdfg1a}
\end{eqnarray}
\begin{eqnarray}
\hspace{-6pt}\mathcal{I}_{t}(\gamma ) \hspace{-6pt}&\triangleq &\hspace{-6pt}\frac{2\overline{\gamma }^{\frac{2}{%
\nu }}\gamma ^{m-1}}{\pi ab\nu \Gamma (m)m^{\frac{2}{\nu }}}\bigintss_{\frac{
md_{(1),t}^{\nu /2}}{\overline{\gamma }}}^{\frac{md_{\mathrm{{max},t}}^{\nu
/2}}{\overline{\gamma }}}\dfrac{x^{m+\frac{2}{\nu }-1}}{\mathrm{e}^{\gamma x}
}\times \\ 
\hspace{-6pt}&&\cos ^{-1}\left( \dfrac{k_{1}+\Lambda _{1}\left( a,b,\sqrt{\left( \frac{x%
\overline{\gamma }}{m}\right) ^{\frac{2}{\nu }}-H^{2}}\right) }{\mathcal{E}%
^{2}\left( \left( \frac{x\overline{\gamma }}{m}\right) ^{\frac{2}{\nu }%
}-H^{2}\right) }\right) \mathrm{d}x, \notag  \label{Int1a}
\end{eqnarray}

\item[$\bullet$] $x_{0} > a$
\begin{eqnarray}
\hspace{-6pt}f_{\textbf{\textgamma}}(\gamma ) \hspace{-6pt}&=&\hspace{-6pt}\frac{2\overline{\gamma }^{\frac{2}{\nu }}\gamma
^{m-1}}{\pi ab\nu \Gamma (m)m^{\frac{2}{\nu }}}\bigintss_{\frac{%
md_{(2),t}^{\nu /2}}{\overline{\gamma }}}^{\frac{md_{\mathrm{{max},t}}^{\nu
/2}}{\overline{\gamma }}}\dfrac{x^{m+\frac{2}{\nu }-1}}{\mathrm{e}^{\gamma x}%
}\times  \label{pdfg1b} \\
&&\cos ^{-1}\left( \dfrac{k_{1}+\Lambda _{1}\left( a,b,\sqrt{\left( \frac{x%
\overline{\gamma }}{m}\right) ^{\frac{2}{\nu }}-H^{2}}\right) }{\mathcal{E}%
^{2}\left( \left( \frac{x\overline{\gamma }}{m}\right) ^{\frac{2}{\nu }%
}-H^{2}\right) }\right) \mathrm{d}x.  \notag
\end{eqnarray}
\end{itemize}
where $\overline{\gamma }\triangleq\frac{P_{t}}{P_{n}}$ is the SNR at the Tx and $
\Gamma (\cdot ,\cdot)$ is the incomplete Gamma function,
as defined in \cite[eq. (06.06.02.0001.01)]{Wolfram}. It is important to note that the integral $\mathcal{I}_t(\gamma)$ in \eqref{pdfg1a} defined in (26) does not have a known analytical solution and must be calculated numerically.

Then, the CDF of the SNR, $\textbf{\textgamma}$, is found after integration as
\begin{itemize}
\item[$\bullet$] $0 < x_{0} \leq a$
\begin{eqnarray}
\hspace{-10pt}F_{\textbf{\textgamma}}(\gamma ) \hspace{-10pt}&=&\hspace{-10pt}\int_{0}^{\gamma }\mathcal{I}_{t}(x)\mathrm{d}x+%
\frac{(a-x_{0})^{2}}{ab}  \label{cdfg1a} \\
&&-\frac{\left( d_{\mathrm{min}}^{2}\mathcal{M}\left( d_{\mathrm{min}%
},\gamma \right) \right) -\left( d_{(1),t}^{2}\mathcal{M}\left(
d_{(1),t},\gamma \right) \right) }{ab\Gamma (m)}.  \notag
\end{eqnarray}

\item[$\bullet$] $x_{0} > a$
\begin{eqnarray}
\hspace{-20pt}F_{\textbf{\textgamma}}(\gamma ) \hspace{-10pt}&=&\hspace{-10pt}\bigintss_{0}^{\gamma }\frac{2\overline{\gamma }^{%
\frac{2}{\nu }}y^{m-1}}{\pi ab\nu \Gamma (m)m^{\frac{2}{\nu }}}\bigintss_{%
\frac{md_{(2),t}^{\nu /2}}{\overline{\gamma }}}^{\frac{md_{\mathrm{{max},t}%
}^{\nu /2}}{\overline{\gamma }}}\dfrac{x^{m+\frac{2}{\nu }-1}}{\mathrm{e}%
^{\gamma y}}\times  \label{cdfg1b} \\
&&\hspace{-20pt}\cos ^{-1}\left( \dfrac{k_{1}+\Lambda _{1}\left( a,b,\sqrt{\left( \frac{x%
\overline{\gamma }}{m}\right) ^{\frac{2}{\nu }}-H^{2}}\right) }{\mathcal{E}%
^{2}\left( \left( \frac{x\overline{\gamma }}{m}\right) ^{\frac{2}{\nu }%
}-H^{2}\right) }\right) \mathrm{d}x\mathrm{d}y.  \notag
\end{eqnarray}
\end{itemize}

\hspace{-7pt} where
\begin{equation}
\hspace{-5pt}\mathcal{M} (x ,\gamma )\hspace{-3pt}\triangleq\hspace{-3pt}\left( \hspace{
-2pt}\dfrac{\gamma m x^{\nu }}{\overline{\gamma }}\hspace{-4pt}\right)
^{-\frac{2}{\nu}}\hspace{-6pt}\Gamma \left( \hspace{-2pt}m\hspace{-3pt}+
\hspace{-3pt}\dfrac{2}{\nu },\dfrac{\gamma m x^{\nu }}{\overline{
\gamma }}\hspace{-2pt}\right) -\Gamma \left( \hspace{-4pt}m,\dfrac{\gamma
m x^{\nu }}{\overline{\gamma }}\hspace{-2pt}\right),  \label{Lambda}
\end{equation}
is a user-defined function, whereas the integrals in the above equations are estimated numerically.

\subsubsection{Second scenario}
Following the same analysis as previously, the PDF of $\textbf{\textgamma}$ is expressed after a bit of algebra in \eqref{pdfg1} where the integral $\mathcal{I}_{\rm{v}}(\gamma)$ defined in \eqref{Int1} is calculated numerically.

\begin{eqnarray}
f_{\textbf{\textgamma}}(\gamma ) &=&\frac{2\overline{\gamma }^{\frac{2}{\nu }}\gamma ^{-%
\frac{\nu +2}{\nu }}}{ab\nu \Gamma (m)m^{\frac{2}{\nu }}}\left( \Gamma
\left( m+\frac{2}{\nu },\frac{\gamma md_{\mathrm{min}}^{\nu }}{\overline{%
\gamma }}\right) \right.   \notag \\
&&-\left. \Gamma \left( m+\frac{2}{\nu },\frac{\gamma md_{(1),\mathrm{v}%
}^{\nu }}{\overline{\gamma }}\right) \right) +\mathcal{I}_{\mathrm{v}%
}(\gamma)  \label{pdfg1}
\end{eqnarray}

\begin{eqnarray}
\mathcal{I}_{\mathrm{v}}(\gamma ) &\triangleq &\frac{4\overline{\gamma }^{%
\frac{2}{\nu }}\gamma ^{m-1}}{\pi ab\nu \Gamma (m)m^{\frac{2}{\nu }}}%
\bigintss_{\frac{md_{(1),\mathrm{v}}^{\nu /2}}{\overline{\gamma }}}^{\frac{%
md_{\mathrm{{max},{v}}}^{\nu /2}}{\overline{\gamma }}}\frac{x^{m+\frac{2}{\nu }-1}}{\mathrm{e}^{\gamma x}}
\times  \notag \label{Int1} \\
&&\cos ^{-1}\left( \sqrt{\frac{(\frac{x\overline{\gamma }}{m})^{\frac{2}{\nu 
}}-d_{(1),\mathrm{v}}^{2}}{\mathcal{E}^{2}\left( (\frac{x\overline{\gamma }}{%
m})^{\frac{2}{\nu }}-H^{2}\right) }}\right) \mathrm{d}x,  
\end{eqnarray}

By integrating \eqref{pdfg1} and after some algebraic manipulations, the CDF of $\textbf{\textgamma}$ is deduced as

\begin{eqnarray}
F_{\textbf{\textgamma}}(\gamma ) &=&\int_{0}^{\gamma }\mathcal{I}_{\mathrm{v}}(x)\mathrm{%
d}x+\frac{b}{a}  \label{cdfg1} \\
&&-\frac{\left( d_{\mathrm{min}}^{2}\mathcal{M}\left( d_{\mathrm{min}%
},\gamma \right) \right) -\left( d_{(1),\mathrm{v}}^{2}\mathcal{M}\left(
d_{(1),\mathrm{v}},\gamma \right) \right) }{ab\Gamma (m)}.  \notag
\end{eqnarray}

\subsection{Outage Performance}
The following subsection provides a comprehensive evaluation of the probability of outage, defined as $P_{\rm{out}}=F_{\textbf{\textgamma}}(\gamma_{\rm{th}})$, for a specific SNR threshold. The default parameters used in the calculations are $\nu=2.5$, $H=300$m and $\overline{\gamma}=95$dB \cite{J:Azari}, \cite{J:Khuwaja}, unless otherwise stated.

\subsubsection{First scenario}

\begin{figure}[h]
\centering
\subfigure[\scriptsize $\gamma_{\rm{th}}=10\rm{dB}$
]{\includegraphics[width=0.4\textwidth]{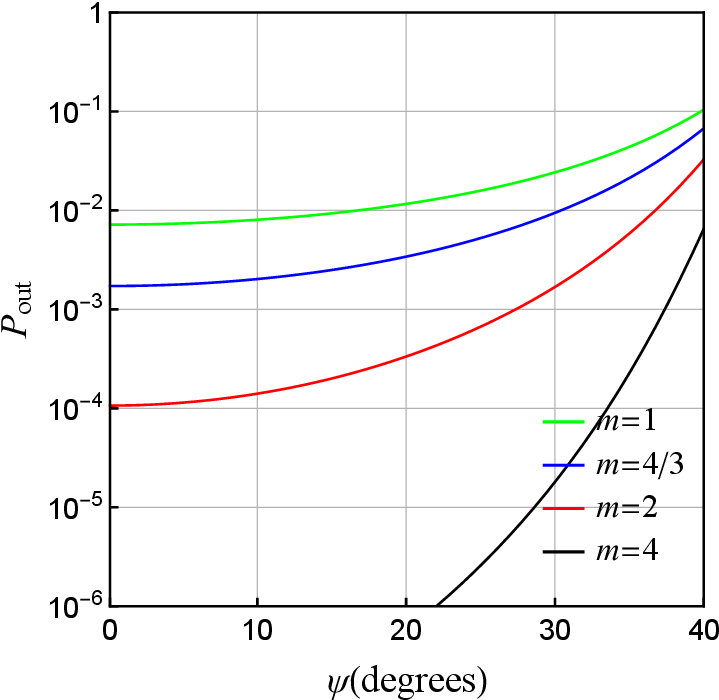}} 
\hspace{0.1in} 
\subfigure[\scriptsize $\gamma_{\rm{th}}=15\rm{dB}$
]{\includegraphics[width=0.4\textwidth]{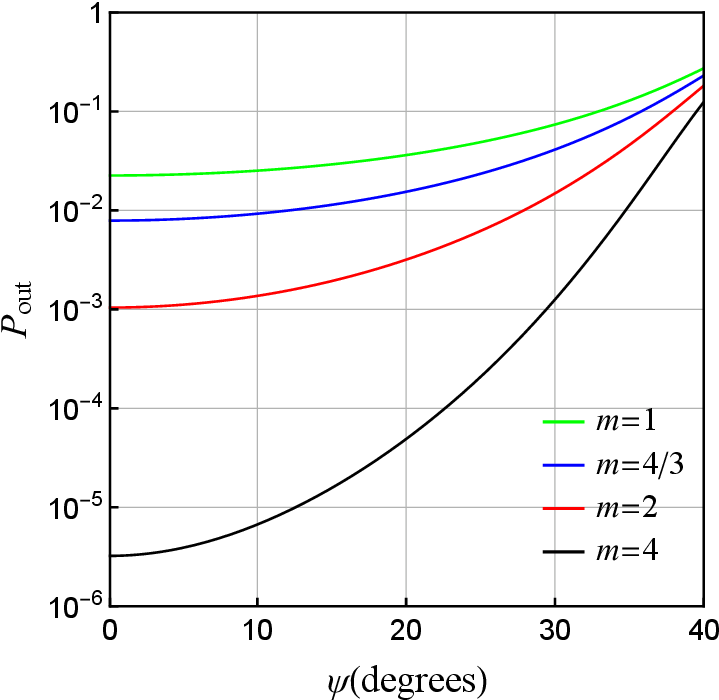}} 
\hspace{0.1in} 
\subfigure[\scriptsize $\gamma_{\rm{th}}=10\rm{dB}$
]{\includegraphics[width=0.4\textwidth]{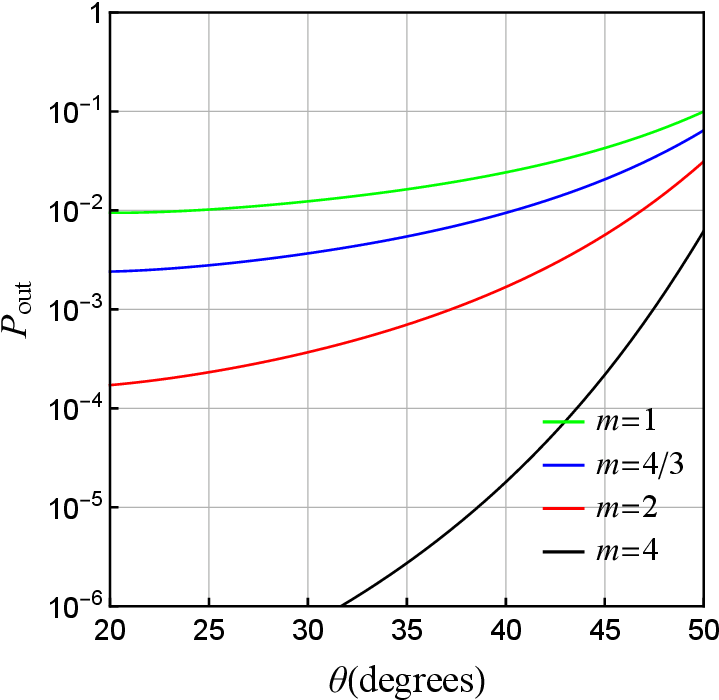}} 
\hspace{0.1in} 
\subfigure[\scriptsize $\gamma_{\rm{th}}=15\rm{dB}$
]{\includegraphics[width=0.4\textwidth]{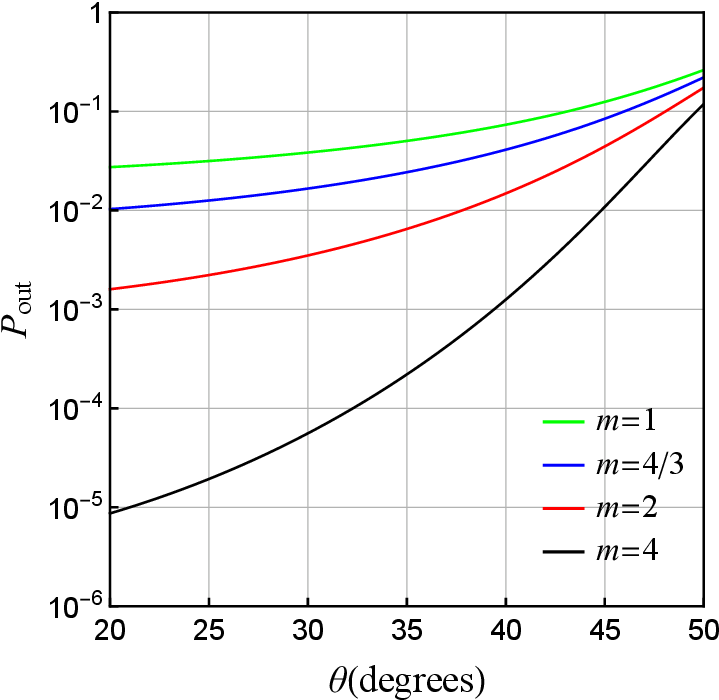}} 
\caption{Outage probability vs. $\psi$ and $\theta$ for two $\gamma_{\rm{th}}$ values ((a), (b) $\theta=40^{\circ}$. (c), (d) $\psi=30^{\circ}$).}
\label{Figure8}
\end{figure}

 Figures \ref{Figure8}(a) through (d) illustrate the impact of the tilt and semi-apex angles, as well as the fading parameter, on the probability of outage. Figures \ref{Figure8}(a) and (b) plot the outage probability against the tilt angle for typical values of $m$, ranging from strict non-line-of-sight (NLOS) conditions ($m=1$) to conditions where the line-of-sight (LOS) component is dominant ($m=4$), assuming $\gamma_{\rm{th}}$ of 10dB and 15dB, respectively, and $\theta=40^{\circ}$. The differences become less significant as the tilt angle increases beyond $40^{\circ}$. However, for lower tilt angles, LOS conditions can significantly improve the coverage area (as shown in Fig. \ref{Figure2}(b)) while also maintaining a low outage probability for both threshold values. This is of considerable importance since there is frequently a requirement to expand the coverage area in response to particular events (such as concerts, sporting events, areas affected by natural disasters, etc.) while upholding a high level of link reliability. Adjusting the tilt of the UAV antenna offers network designers a valuable alternative in pursuing this objective. Similar findings can be observed in Figs. \ref{Figure8}(c) and (d), where the tilt angle is fixed at $30^{\circ}$, and $\theta$ varies for different values of $m$. Again, the strong LOS component allows low levels of outage probability at higher $\theta$ values. For example, assuming an acceptable $P_{\rm{out}}$ of $10^{-4}$, the angle $\theta$ can be increased to $44^{\circ}$ for a threshold of 10dB or $39^{\circ}$ for a threshold of 15dB. As a result, the semi-apex angle can provide network designers with additional freedom to improve coverage in tilted transmissions while maintaining link quality.

\begin{figure}[h]
\centering
\subfigure[\scriptsize $\gamma_{\rm{th}}=10\rm{dB}$
]{\includegraphics[width=0.4\textwidth]{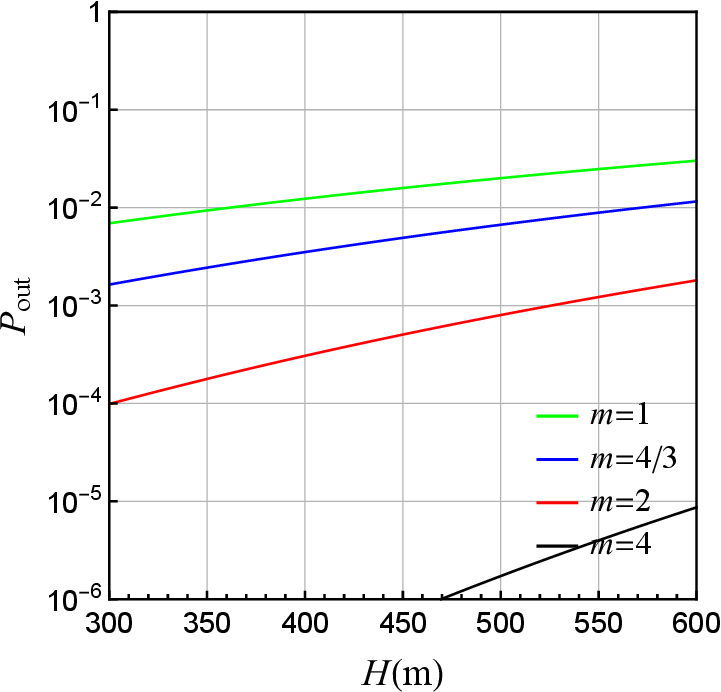}} 
\hspace{0.1in} 
\subfigure[\scriptsize $\gamma_{\rm{th}}=15\rm{dB}$
]{\includegraphics[width=0.4\textwidth]{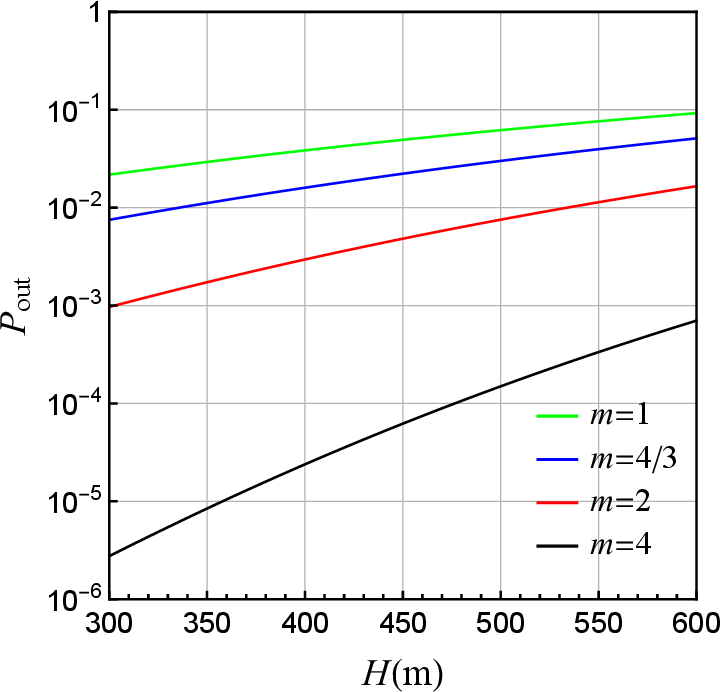}}
\hspace{0.1in} 
\subfigure[\scriptsize $\gamma_{\rm{th}}=10\rm{dB}$, $H=450m$
]{\includegraphics[width=0.4\textwidth]{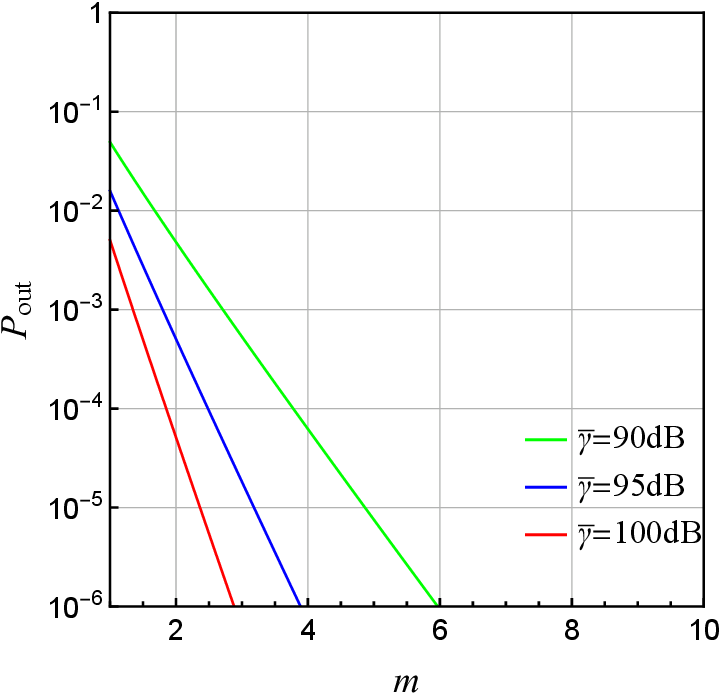}}
\hspace{0.1in} 
\subfigure[\scriptsize $\gamma_{\rm{th}}=15\rm{dB}$, $H=450m$
]{\includegraphics[width=0.4\textwidth]{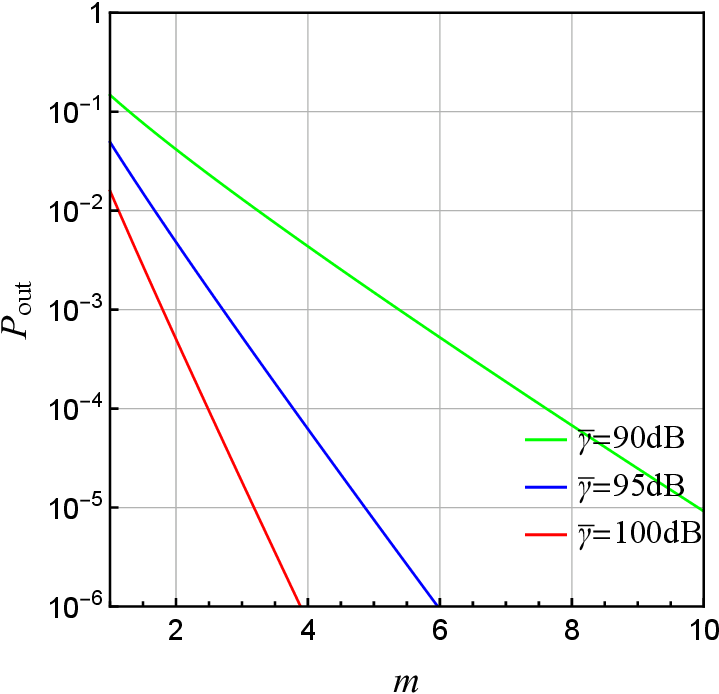}}
\caption{Outage probability vs (a), (b) $H$ and (c), (d) $m$ for two $\gamma_{\rm{th}}$ values.}
\label{Figure9}
\end{figure}

Table \ref{Table2} summarizes the key performance trends observed in various combinations of UAV altitude and antenna tilt angle for the first scenario, assuming a threshold SNR of $\gamma_{\rm{th}} = 10$dB. Each column in the table corresponds to a specific case, characterized by a distinct set of parameters, such as the Nakagami‑$m$ fading parameter, target outage probability, and antenna semi-apex angle.
\begin{itemize}
    \item Case 1: $m=1$, $P_{\rm{out}}=10^{-2}, \theta=30^{\circ}$
    \item Case 2: $m=1$, $P_{\rm{out}}=10^{-2}, \theta=40^{\circ}$
    \item Case 3: $m=2$, $P_{\rm{out}}=10^{-4}, \theta=30^{\circ}$
    \item Case 4: $m=2$, $P_{\rm{out}}=10^{-4}, \theta=40^{\circ}$
    \item Case 5: $m=4$, $P_{\rm{out}}=10^{-7}, \theta=30^{\circ}$
    \item Case 6: $m=4$, $P_{\rm{out}}=10^{-7}, \theta=40^{\circ}$
\end{itemize}

\begin{table}[ht]
 
\centering
\caption{First scenario: Tilt angle, $\psi$, vs. altitude, $H$.}
\label{Table2}
\begin{tabular}{c||cccccc}
\hline \hline
\multirow{2}{*}{$\psi(^o)$} & \multicolumn{6}{c}{$H$(m)} \\
\cline{2-7}
 & Case 1 & Case 2 & Case 3 & Case 4 & Case 5 & Case 6 \\
\hline
10 & 358.7 & 327.3 & 310.3 & 280.1 & 352.6 & 309.3\\
15 & 345.7 & 308.4 & 297.4 & 260.9 & 332.3 & 297.2\\
20 & 327.4 & 282.3 & 279.7 & 235.3 & 309.5 & 250.6\\
25 & 304.0 & 249.2 & 257.3 &204.3 & 285.0 & 215.6\\
\hline \hline
\end{tabular}
\end{table}

\subsubsection{Second scenario}
Figures \ref{Figure9}(a) to (d) illustrate the impact of altitude and fading parameters on the probability of outage with $a=259.8$m and $b=212.2$m. It is evident from Figs. \ref{Figure9}(a) and (b) that achieving a target outage probability of $10^{-4}$ is feasible at higher altitudes ($H=480$m for $m=4$ and $\gamma_{\rm{th}}=15$dB) under LOS conditions. In contrast, Figs. \ref{Figure9}(c) and (d) suggest that the influence of the LOS component diminishes with increasing $\overline{\gamma}$, as the target value is reached at lower values of $m$. 

\begin{figure}[h]
\centering
\includegraphics[width=4in]{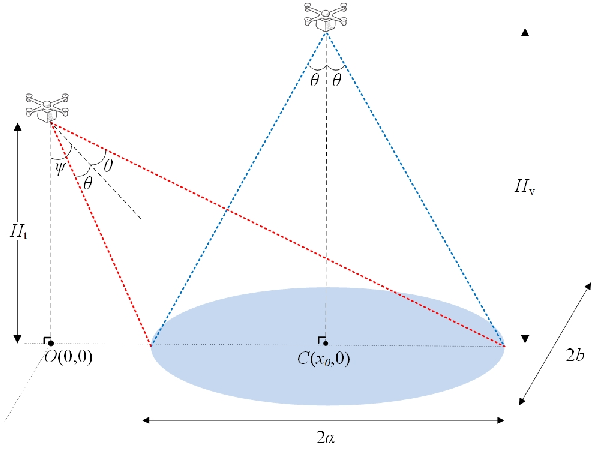}
\caption{Tilted vs. vertical transmission.}
\label{Figure10}
\end{figure}

\subsubsection{Comparison of two scenarios} To perform a valid comparison between the two scenarios, we utilize a UAV with a tilted directional antenna and a second UAV with a vertical directional antenna, both of which cover the same elliptical area on the ground, as illustrated in Fig. \ref{Figure10}. The first UAV hovers at $H_{t}=300$m and has a directional antenna tilted at $\psi$ degrees to cover the shaded area. Meanwhile, the second UAV covers the same area but hovers at altitude $H_{\rm{v}}$. In both cases, the semi-apex angle is $\theta=30^{\circ}$, while $m=4/3$. Furthermore, according to \eqref{EllDim}, the values of the axes $\{a,b\}$ of the ellipse for $\psi=\{20^{\circ},40^{\circ}\}$ are $\{205.2, 188.5\}$ and $\{385.6, 258.5\}$, respectively.

\begin{figure}[h]
\centering
\subfigure[\scriptsize ]{\includegraphics[width=0.4\textwidth]{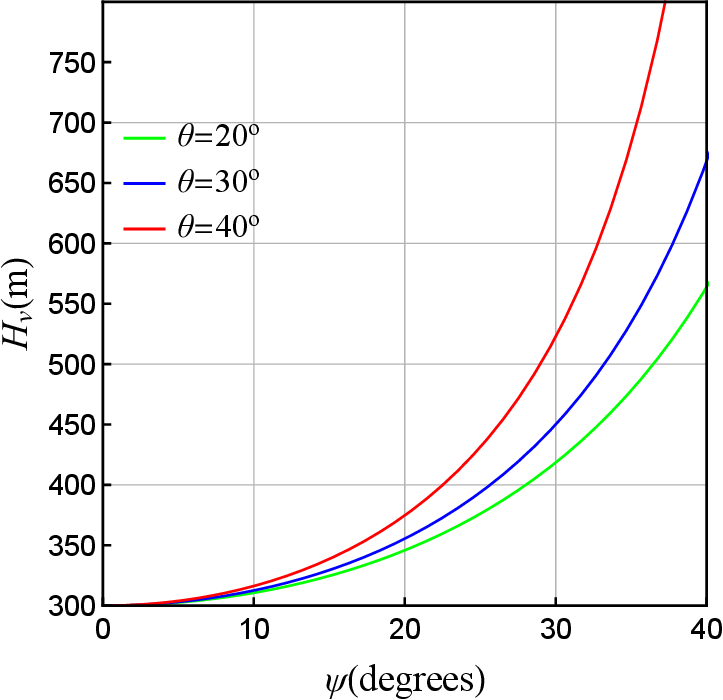}} 
\hspace{0.1in} 
\subfigure[\scriptsize $H_{\rm{v}}^{(1)}=355.4m, H_{\rm{v}}^{(2)}=668m$]{\includegraphics[width=0.4\textwidth]{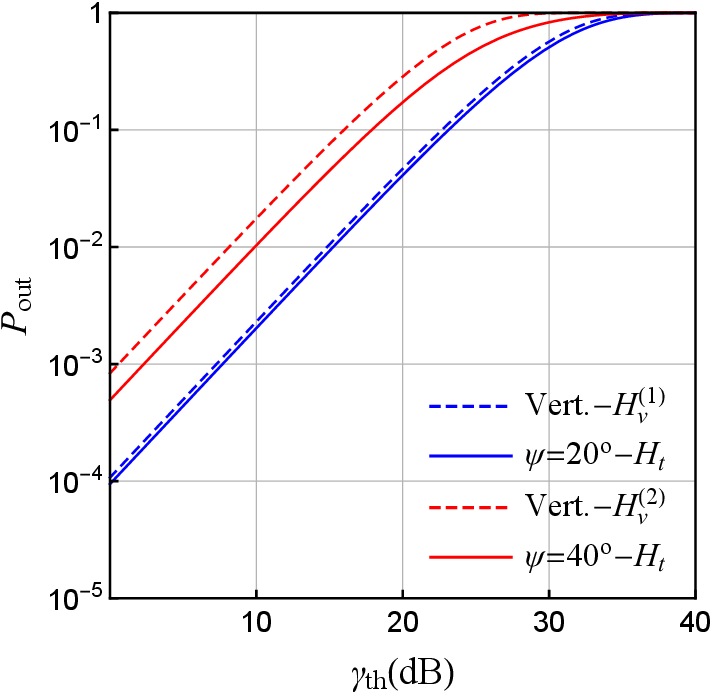}}
\caption{(a) Required altitude for vertical transmission vs. tilt angle, assuming $m=4/3$. (b) Outage probability vs. $\gamma_{\rm{th}}$ for both scenarios and two tilt angles.}
\label{Figure11}
\end{figure}

Figure 11 outlines the performance comparison between the two UAVs. Within the context of the second scenario, Fig. \ref{Figure11}(a) displays the association between $\psi$ and $H_{\rm{v}}$, required by the UAV for directional transmission. The represented tradeoff implies that for the same area coverage, an increase in tilt angle (as per the first scenario) leads to a higher requirement for $H_{\rm{v}}$. Furthermore, as $\theta$ increases, the required altitude increases even further.

The UAVs are compared within the same elliptical region for two separate $\psi$ values, namely $\psi=20^{\circ}$ (blue line) and $\psi=40^{\circ}$ (red line) for the first UAV hovering at $H_t=300$m, as shown in Fig. \ref{Figure11} (b). Furthermore, the corresponding curves for the second UAV with vertical directional transmission are displayed by dashed lines with the altitude values extracted from Fig. \ref{Figure11}(a). The analysis indicates that the use of antenna tilt results in a marginal decrease in the probability of outage at $\psi=20^{\circ}$ for a specific threshold of $\gamma_{th}$, which increases when $\psi$ increases to $40^{\circ}$. This disparity arises because the second scenario requires a greater altitude ($H_{\rm{v}}^{(1)}=355.4$m for $\psi=20^{\circ}$ and $H_{\rm{v}}^{(2)}=668$m for $\psi=40^{\circ}$) to cover the same area. Based on the above, it can be argued that implementing the first scenario can enhance performance and energy efficiency since UAVs with vertical directional antennas need to ascend to significantly higher altitudes to cover the same footprint areas.

\subsection{MU Scenario}
Subsequently, we explore an MU scenario in which $M$ Rx are uniformly distributed within an elliptical disc. To address this scenario, we employ a representative scheduling scheme called maxSNR. This scheme allocates additional resources to Rxs based on the optimal channel state observed during each transmission time. The CDF and PDF of $\textbf{\textgamma}$ are subsequently formulated, as outlined in \cite{J:Vaiopoulos2},

\begin{equation}
F_{\textbf{\textgamma},MU}(\gamma)=(F_{\textbf{\textgamma}}(\gamma))^{M}, \label{cdfSNRsch}
\end{equation}

\begin{equation}
f_{\textbf{\textgamma},MU}(\gamma)=M(F_{\textbf{\textgamma}}(\gamma))^{M-1}f_{\textbf{\textgamma}}(\gamma), \label{pdfSNRsch}
\end{equation}
whereas the throughput can be computed as

\begin{equation}
\overline{C}_{MU}=\int_{0}^{\infty}\log_{2}(1+\gamma)M(F_{\textbf{\textgamma}}(\gamma))^{M-1}f_{\textbf{\textgamma}}(\gamma)d\gamma. \label{ThrSNR}
\end{equation}
Substituting the expressions \eqref{pdfg1a} or \eqref{pdfg1b} and \eqref{cdfg1a} or \eqref{cdfg1b}, relevant to scenario 1, or utilizing \eqref{pdfg1} and \eqref{cdfg1} linked to scenario 2, in the framework of equations \eqref{cdfSNRsch} and \eqref{pdfSNRsch}, respectively, allows calculation of the numerical value for \eqref{ThrSNR}.

\begin{figure}[h]
\centering
\subfigure[\scriptsize $H=300m, \theta=40^{\circ}$]{\includegraphics[width=0.4\textwidth]{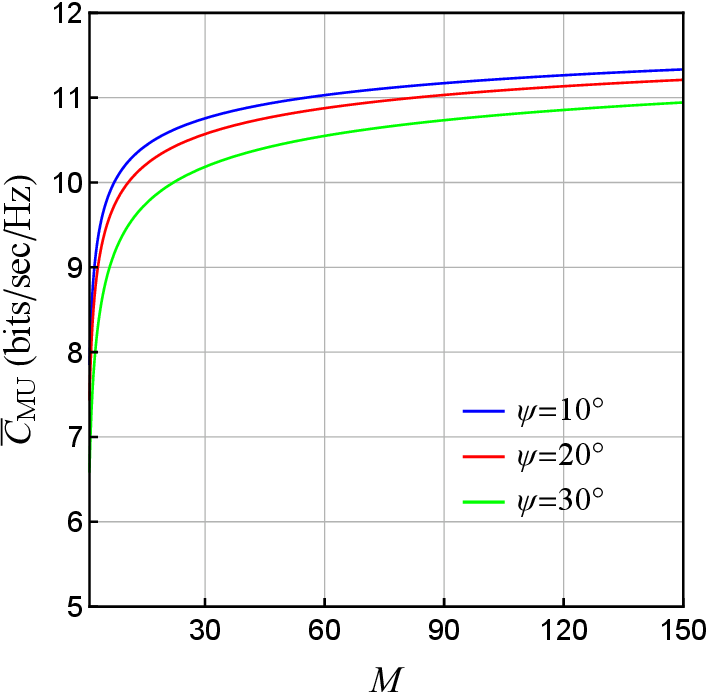}} 
\hspace{0.1in} 
\subfigure[\scriptsize $a=259.8m, b=212.1m$]{\includegraphics[width=0.4\textwidth]{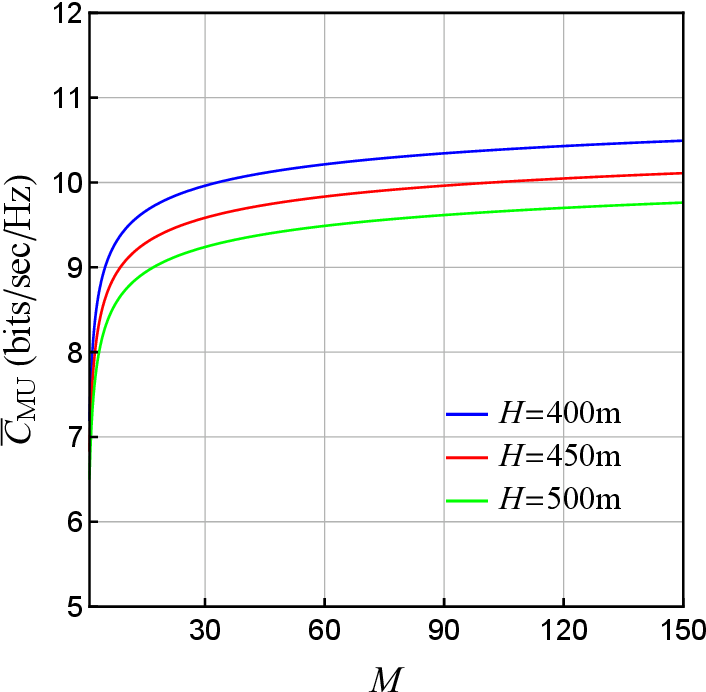}}
\caption{Throughput results for (a) scenario 1 and (b) scenario 2.}
\label{Figure12}
\end{figure}

Figure \ref{Figure12} visualizes the resulting throughput for both scenarios as a function of $M$, considering $m=1$, $\overline{\gamma}=90$dB and varying values of the tilt angle, $\psi$, for the first scenario and height, $H$, for the second scenario. The throughput experiences a notable increase with lower tilting angles as more Rxs are closer to the Tx, facilitating the reception of higher power. A similar trend is observed in the second scenario when lower altitude values are used.

\begin{figure}[h]

\centering
{\includegraphics[width=0.4\textwidth]{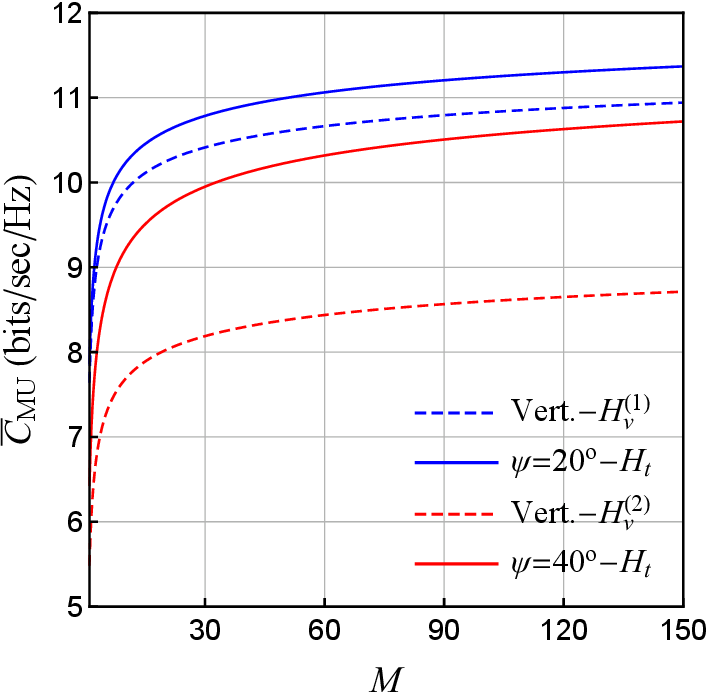}}
\caption{Throughput comparison between the two scenarios ($H_{\rm{v}}^{(1)}=355.4$m, $H_{\rm{v}}^{(2)}=668$m).}
\label{Figure13}
\end{figure}

Figure \ref{Figure13} presents a throughput comparison for both scenarios as a function of $M$, using the same elliptical footprints and configuration parameters as in Fig. \ref{Figure11}. The analysis indicates that antenna tilt results in a slight increase in throughput at $\psi=20^{\circ}$ for a specific $\gamma_{th}$, which becomes substantial when $\psi$ reaches $40^{\circ}$. These findings confirm the trends observed in Fig. \ref{Figure11}.

 The results also indicate that as the number of users increases, the overall performance gain tends to saturate, primarily because the available resources are preferentially allocated to users experiencing higher SNR. Consequently, the benefit of adding more users diminishes, reflecting inherent limitations in resource fairness and efficiency in dense user deployments. 
 
The number of users is a critical parameter of the system, and high-density MU scenarios introduce several key challenges. Although the derived analytical expressions remain applicable as the system scales, their numerical evaluation becomes increasingly computationally intensive with a growing user population. To address this, future research may focus on developing efficient approximation techniques that maintain analytical tractability while significantly reducing computational complexity.  

\subsection{Comparative Analysis of Footprint Models}
We consider three representative footprint geometries: (i) a circular disc of radius $R$, (ii) a hexagon inscribed within the disc, and (iii) an elliptical disc with major semi-axis $a=R$, assuming a UAV at altitude $H$ with vertical orientation. The circular and hexagonal cases employ identical HPBWs, whereas the elliptical disc follows Scenario 2. The circular footprint is obtained as a special case of the elliptical configuration. For hexagonal geometry, the CDF of $\mathbf{r}$ is given in (19) of \cite{J:Khalid13}, while $\mathbf{d}$ and $\boldsymbol{\gamma}$ can be expressed in closed form using (12) and (24), respectively. Figure \ref{Figure14} illustrates a comparison of these footprints for $R=350$m, $H=200$m, $\nu=2.5$, $m=4/3$, and $\overline{\gamma}=95$dB. The results demonstrate that the elliptical footprint achieves lower $P_{out}$ at low-to-moderate SNR (Fig. 14(a)), and superior throughput scalability in MU scenarios (Fig. 14(b)), thus providing a better trade-off between coverage and diversity.

\begin{figure}[h]
\centering
\subfigure[]{\includegraphics[width=0.4\textwidth]{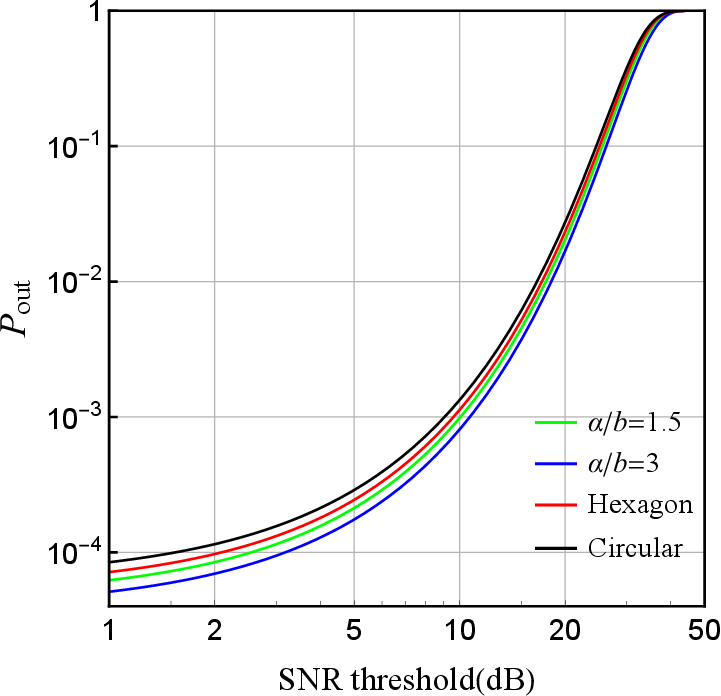}} 
\hspace{0in} 
\subfigure[]{\includegraphics[width=0.4\textwidth]{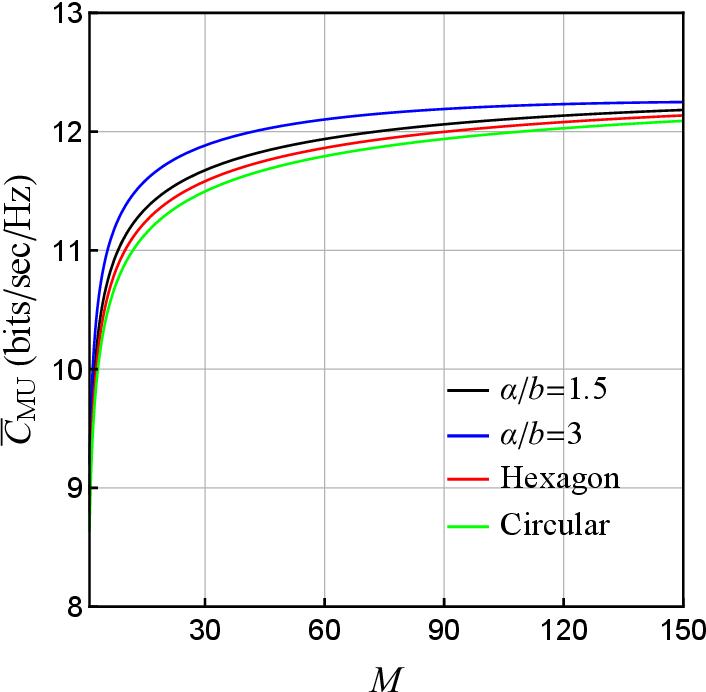}} 
\caption{Comparison of footprint geometries: (a) outage probability vs. SNR threshold, and (b) throughput vs. number of users $M$.}
\label{Figure14}
\end{figure}

\section{Discussion}

\subsection{Concluding Remarks}
In conclusion, current research provides valuable insight into the performance of wireless systems with random terminal positions within an elliptical area. The study focused on an air-to-ground outdoor link between a UAV and a randomly placed user and evaluated the outage performance, including Nakagami-$m$ fading. Two distinct scenarios were considered. The first assumed an airborne Tx with a tilted directional antenna, whereas the second had a vertical directional antenna. We illustrate the coverage of an area implemented employing both scenarios and conducted a performance comparison between them. Our investigation demonstrated that the use of tilt transmission leads to a reduced probability of outage compared to vertical transmission within the same coverage region. Furthermore, modifying the angle to achieve a horizontal displacement is a standard and more energy efficient method to extend the coverage area compared to the vertical elevation necessary in the second scenario. Ultimately, the scenario of MU was considered and throughput was evaluated. 

 The analytical results offer practical value for the design of UAV-based wireless systems. The derived distance and SNR distributions enable the evaluation of the coverage at the system level, helping to optimize tilt angles and beamwidths to reduce outage or boost throughput. The closed-form expressions also support adaptive UAV placement with lower altitudes and energy use. In MU settings, metrics guide scheduling strategies such as maxSNR-based allocation to improve spectral efficiency.  

The study can help designers determine the desired elliptical footprints based on the type of antenna available. The equations deduced for the radial and Euclidean distances between the UAV and the user play a vital role in comprehensively assessing the network performance. These equations provide a means to explore other network metrics, including, but not limited to, bit error rate, capacity, and energy efficiency. Moreover, they can assist in optimizing various network parameters, such as the Tx antenna pattern or the throughput in a MU scenario.

  In real-world applications such as emergency response, post-disaster connectivity restoration, and temporary event coverage, UAVs are frequently deployed in complex terrains and over non-uniform user distributions. The flexibility enabled by elliptical coverage, achieved through adjustment of the tilt of the antenna or the shaping of the beam, can play a crucial role in enhancing communication performance under such dynamic conditions. Future research may investigate the integration of these models into real-time UAV path planning and coverage optimization frameworks.

In practical scenarios --particularly in urban environments-- user distributions are typically non-uniform and may exhibit spatial clustering due to factors such as population density, infrastructure layout, or event-driven gatherings, for e.g. sports or music events. To address such cases, the geometric centroid of the user cluster can serve as the center of an adapted elliptical coverage footprint. Furthermore, statistical techniques can be employed to estimate the orientation and axis lengths of the ellipse that best represent the spatial distribution.  

\subsection{Further Research}
This paper lays the groundwork for extensive future research by building on the foundational assumptions. The analysis establishes a fundamental reference point by assuming equiprobable user locations within the elliptical disk. This baseline is a benchmark for exploring more specialized scenarios where specific locations may exhibit a higher probability of user presence. 

 Although the derived PDF and CDF are in closed form, their evaluation in dense MU settings can be computationally demanding due to numerical integration. Future work may employ high-density approximations by modeling users as a continuous spatial distribution, enabling simplified analysis. For example, with uniformly distributed users, the mean path loss over the elliptical region can approximate the SNR via the law of large numbers. Surrogate methods such as moment-based models or pre-computed lookup tables can further reduce real-time complexity, supporting scalable UAV network design. 

Another critical assumption concerns the consistent antenna gain throughout the coverage area. The potential variability of antenna gain with elevation and/or azimuth angles, as discussed in \cite{J:Kim22} and \cite{J:wang22}, remains pertinent and applicable in the context of VLC, where channel gain adheres to a Lambertian model \cite{J:Vaiopoulos2}.

 This work assumes a static elliptical footprint, but real UAV deployments face the dynamics of motion, trajectory shifts, and environmental effects such as wind or shadowing. These factors cause variation in footprint and user mobility. Future work can model time-varying ellipses by updating parameters such as tilt angle, beamwidths, and center position, enabling dynamic SNR and outage analysis. Incorporating user mobility models and real-time UAV control (e.g., repositioning, beam steering) would enhance the realism and applicability of the model in dynamic UAV networks \cite{J:Aalo}.  

\bibliographystyle{IEEEtran}
\bibliography{IEEEabrv,References}

\balance
\end{document}